\documentclass{amsart}
\usepackage[utf8]{inputenc}
\usepackage[margin=1.1in]{geometry}
\usepackage{amsfonts}
\usepackage{amssymb}
\usepackage{amsmath}
\usepackage{amsthm}
\usepackage{mathtools}
\usepackage{mathrsfs}

\usepackage{enumerate}
\usepackage{braket}
\usepackage{comment}
\usepackage{cases}
\usepackage{multirow}
\usepackage{enumitem}

\usepackage{hyperref}

\usepackage{graphicx}
\usepackage{tikz}
\usetikzlibrary{matrix}
\usepackage{caption}
\usepackage{subcaption}

\usepackage{xcolor}
\usepackage{cancel}
\usepackage{ulem}

\newtheorem{lemma}{Lemma}
\newtheorem{corollary}{Corollary}
\newtheorem{remark}{Remark}
\newtheorem{proposition}{Proposition}
\newtheorem{definition}{Definition}
\newtheorem{theorem}{Theorem}

\thanks{The author is much obliged to Matei Machedon and Manoussos Grillakis for proposing the problem and their enlightening discussions and advice.}
\title[Well-posedness for Hartree Equation]{The Hartree equation with a constant magnetic field: Well-posedness theory}
\author{Xin Dong}
\date{\today}

\begin{document}

\begin{abstract}
    We consider the Hartree equation for infinitely many electrons with a constant external magnetic field. For the system, we show a local well-posedness result when the initial data is the pertubation of a Fermi sea, which is a non-trace class stationary solution to the system. In this case, the one particle Hamiltonian is the Pauli operator, which possesses distinct properties from the Laplace operator, for example, it has a discrete spectrum and infinite-dimensional eigenspaces. The new ingredient is that we use the Fourier-Wigner transform and the asymptotic properties of associated Laguerre polynomials to derive a collapsing estimate, by which we establish the local well-posedness result.
\end{abstract}
\maketitle

\section{Introduction}

For a system of $N$ electrons moving in a constant magnetic field $B=(0,0,b)$, $(b>0)$, the Hamiltonian is described by 
\begin{equation}
    \hat{\mathcal{H}}_{N}=\sum_{j=1}^{N}h_{j}+\sum_{j>k}^{N}w\left(x_{j}-x_{k}\right),\quad x_{j}\in\mathbb{R}^3,
\end{equation}
and the Schr\"odinger equation is 
\begin{equation}\label{many,body,schrodinger}
    i\,\partial_{t}\Psi_{N}(t,x_{1},x_2,\ldots,x_N)=\hat{\mathcal{H}}_{N}\Psi_{N}(t,x_{1},x_2,\ldots,x_N),\quad \Psi_{N}(t=0)=\Psi_{N,0}\in \wedge^{N}L^2\left(\mathbb{R}^3,\mathbb{C}^2\right),
\end{equation}
where $h=\left(\sigma \cdot (-i\nabla -A)\right)^2$ is the Pauli operator 
\begin{enumerate}
        \item[(a)] $\sigma=(\sigma_{1},\sigma_{2},\sigma_{3})$ are Pauli matrices ;
        \item[(b)] $\displaystyle A=-\frac{b}{2}\left(x^2,-x^1,0\right)$ \footnote{There are other choices of $A$, for example $A=-b(x^2,0,0)$ \cite[Chapter XV]{LandauLifVol3}. We use the one which is fixed by the Coulomb gauge $\nabla\cdot A=0$.} is the vector potential of the field $\mathbf{B}=\nabla\times A$,
\end{enumerate}
$h_{j}$ means $h$ acts on the variable $x_{j}$ (the $j$-th electron) and $w$ is the pairwise interaction potential. By the Pauli exclusion principle, $\Psi_{N}$ is in the space $\wedge^{N}L^2\left(\mathbb{R}^3,\mathbb{C}^2\right)$ of anti-symmetric functions. A direct computation shows
\[
    h=\begin{pmatrix} \left(-i\nabla-A\right)^2 & 0\\ 0 & \left(-i\nabla-A\right)^2\end{pmatrix}-\sigma\cdot\mathbf{B},
\]
while $\sigma\cdot\mathbf{B}=\begin{pmatrix} b & 0\\ 0 & -b\end{pmatrix}$ is harmless for the analysis of the system. For simplicity, we consider the scalar case, i.e. $h=\left(-i\nabla-A\right)^2$.

The initial data $\Psi_{N,0}$ is set to be a Slater determinant 
\begin{align*}
    \Psi_{N,0}(x_{1},x_2,\ldots,x_{N})&=\psi_{1,0}\wedge\psi_{2,0}\wedge\cdots\wedge\psi_{N,0}(x_{1},\ldots,x_{N})\\
    &:=\frac{1}{\sqrt{N!}}\sum_{\pi\in S_{N}}sgn(\pi) \psi_{1,0}(x_{\pi(1)}) \psi_{2,0}(x_{\pi(2)})\cdots \psi_{N,0} (x_{\pi(N)})
\end{align*}
with a family of $N$ orthonormal orbitals $\{\psi_{j,0}\}_{j=1}^{N}$ in $L^2(\mathbb{R}^3)$ ($sgn(\pi)$ denotes the sign
of the permutation $\pi\in S_{N}$). $\Psi_{N,0}$ is presumed as an approximation to the ground state of $\hat{\mathcal{H}}_{N}$. The total energy of Equation (\ref{many,body,schrodinger}) is $\left\langle\Psi_{N},\hat{\mathcal{H}}_{N}\Psi_{N}\right\rangle$. At the initial time $t=0$, a direct computation shows that the energy is in the following form
\begin{align}
    \sum_{j=1}^{N}\left\langle \psi_{j,0},h\psi_{j,0}\right\rangle+\frac{1}{2}\int_{\mathbb{R}^3\times\mathbb{R}^3}\left(\sum_{j=1}^{N}|\psi_{j,0}|^2\right)(x)w(x-y)\left(\sum_{j=1}^{N}|\psi_{j,0}|^2\right)(y)\,dxdy\label{slater,many,body}\\
    -\frac{1}{2}\int_{\mathbb{R}^3\times\mathbb{R}^3}\left(\sum_{j=1}^{N}\psi_{j,0}(x)\bar{\psi}_{j,0}(y)\right)^2w(x-y)\,dxdy.\nonumber
\end{align}

After the time evolution, $\Psi_{N}$ may not necessarily stay as a Slater determinant. Instead, one might expect that in an appropriate sense,
\[
    \Psi_{N}(t,x_1,\ldots,x_N) \approx \left(\psi_1(t)\wedge\psi_2(t)\wedge\cdots\wedge\psi_N(t)\right)(t,x_1,\ldots,x_N),
\]
for short time. While $\psi_{j}(t)$ is described by the following  Hartree-Fock equations, for $j=1,\ldots,N$,
\begin{align}\label{hartree,fock,N}
    \left\{ \begin{array}{l}
    \displaystyle i\,\partial_{t}\psi_{j}(t,x)=\left(h+\rho_{t}*w-X\right)\psi_{j}(t,x), \\
    \displaystyle \psi_{j}(0,x)=\phi_{j,0}(x),\quad \psi_{j,0}\in L^2\left(\mathbb{R}^3\right), 
    \end{array}\right. t\in\mathbb{R},\ x\in\mathbb{R}^3.
\end{align}
where $\rho_{t}(x)$ is the particle density 
\[
    \rho_{t}(x)=\sum_{k=1}^{N}\left|\psi_{k}\right|^2(t,x),
\]
$\rho_{t}*w$ denotes the usual convolution
\[
    \left(\rho_{t}*w\right)(x):=\int_{\mathbb{R}^3}\rho_{t}(x-y)w(y)\,dy,
\]
and $X$ is the integral operator with kernel
\[
    X(t,x,y)=\sum_{k=1}^{N}w(x-y)\psi_{k}(t,x)\bar{\psi}_{k}(t,y).
\]
$\{\psi_{j}(t,x)\}_{j=1}^{N}$ remains an orthonormal set as long as Equations (\ref{hartree,fock,N}) are well-posed. Equations (\ref{hartree,fock,N}) come with a total energy
\begin{align*}
    \sum_{j=1}^{N}\left\langle \psi_{j},h\psi_{j}\right\rangle+\frac{1}{2}\int_{\mathbb{R}^3\times\mathbb{R}^3}\rho_{t}(x)w(x-y)\rho_{t}(y)\,dxdy
    -\frac{1}{2}\int_{\mathbb{R}^3\times\mathbb{R}^3}\left(\sum_{j=1}^{N}\psi_{j}(x)\bar{\psi}_{j}(y)\right)^2w(x-y)\,dxdy.\nonumber,
\end{align*}
which assumes the same expression as (\ref{slater,many,body}) at the initial time.

In a mean field regime and in the absence of the magnetic field, i.e. $h=-\Delta$, with a scaling of the kinetic part and the interaction part, Equations (\ref{hartree,fock,N}) are an effective description of Equation (\ref{many,body,schrodinger}) for certain $w$ and initial data, when $N$ is sufficiently large. See details in \cite{BPS2014}. In \cite{BPS2014}, the exchange term $X$ is of lower order and they also proved that the effective description remains true if Equations (\ref{hartree,fock,N}) are replaced by the following $N$ Hartree equations \footnote{They are called Hartree equations since the operator $h+\rho_{t}*w$ is derived by applying the variational principle to the Hartree product  $\psi_{1}\otimes\cdots\otimes\psi_{N}$ instead of the Slater determinant $\psi_{1}\wedge\cdots\wedge\psi_{N}$ \cite[Chapter Three]{MQC96}.} in the reduced Hartree-Fock \cite{Solovej1991} model, for $j=1,\ldots,N$, 
\begin{align}\label{N,hartree,equations}
    \left\{\begin{array}{l}
        \displaystyle i\,\partial_{t}\psi_{j}(t,x)=\left(h+\rho_{t}*w\right)\psi_{j}(t,x),\\
        \displaystyle \psi_{j}(0,x)=\psi_{j,0}(x),\quad \psi_{j,0}\in L^2\left(\mathbb{R}^3\right), 
    \end{array}\right. t\in\mathbb{R},\ x\in\mathbb{R}^3.
\end{align}   
We refer to \cite{BGGAM2003,EESY04,FW11} for other comparisons on the three dynamics from a perspective of mean field and semi-classical limit and refer to \cite{NS1981,Spohn1981} for a different mean field limit of Equation (\ref{many,body,schrodinger}) on the Vlasov hierarchy.

The problem of our interest is the well-posedness theory of Hartree equations (\ref{N,hartree,equations}) when we take the formal limit of $N$ to be infinite.  In order to give a mathematical description of the problem, we adopt the density matrix formulation of the Hartree equations. Denote the density matrix associated to $\{\psi_{j}\}_{j=1}^{N}$ by 
\begin{equation}\label{kernel,Gamma}
    \Gamma_{N}(t,x,y)=\sum_{j=1}^{N}\psi_{j}(t,x)\bar{\psi}_{j}(t,y),\quad x,y\in\mathbb{R}^3,
\end{equation}
where $\Gamma_{N}$ can be thought as an operator from $L^2\left(\mathbb{R}^3\right)$ to itself and (\ref{kernel,Gamma}) is the expression for the integral kernel of $\Gamma_{N}$, i.e.
\begin{equation}
    \left(\Gamma_{N}f\right)(x)=\int_{\mathbb{R}^{3}}\Gamma_{N}(t,x,y)f(y)\,dy,\quad x\in\mathbb{R}^3.
\end{equation}
If operators have integral kernels, for simplicity, we use the same notations for the operators and their integral kernels. Based on Hartree equations (\ref{N,hartree,equations}),  $\Gamma_{N}$ satisfies the following operator equation
\begin{align}
    \left\{\begin{array}{l}
    \displaystyle i\,\partial_{t}\Gamma_{N} =\left[h+\rho_{\Gamma_{N}}*w,\Gamma_{N}\right],\\ 
    \displaystyle \Gamma_{N}(0,x,y)=\Gamma_{N0}(x,y)=\sum_{j=1}^{N}\psi_{j,0}(x)\bar{\psi}_{j,0}(y),\quad x,y\in\mathbb{R}^3
    \end{array}\right.
\end{align}
where $\rho_{\Gamma_{N}}(t,x)=\Gamma_{N}(t,x,x)$, $[A,B]=AB-BA$ and $\rho_{\Gamma_{N}}*w$ denotes the multiplication operator on $L^2\left(\mathbb{R}^3\right)$ by $\rho_{\Gamma_{N}}*w$. 

Note that $\left\|\Gamma_{N0}\right\|_{tr}=N$. In this formulation, as $N\rightarrow \infty$, the trace norm of $\Gamma_{N0}$ blows up.
Therefore the case we want to study is the following Hartree equation
\begin{align}\label{hatree,density,form}
    \left\{\begin{array}{l}
        i\,\partial_{t}\Gamma =\left[h+\rho_{\Gamma}*w,\Gamma\right],\\
        \Gamma(t=0)=\Gamma_{0}, 
    \end{array}\right.
\end{align}
where $\rho_{\Gamma}(t,x)=\Gamma(t,x,x)$, with $\Gamma_{0}$ not being of trace class. Notice that the Pauli exclusion principle of infinitely many electrons requires that $\Gamma_{0}$ satisfies the operator inequality $0\le\Gamma_{0}\le1$.

In the absence of magnetic fields, i.e. $h=-\Delta$, if $\Gamma_{0}$ is not of trace class, Equation (\ref{hatree,density,form}) was recently studied by several authors \cite{LS2015,LS2014,CHP2017,CHP2018} and they showed global well-posedness and the long time scattering behavior separately for different interaction potentials $w$; if $\Gamma_{0}$ is of trace class, the Hartree-Fock equation \footnote{In this case, the methods for the Hartree-Fock equation can be directly applied to the Hartree equation.} (adding the exchange term to Equation (\ref{hatree,density,form})) has been treated by \cite{BPF1974,BPF1976,Cha1976,Zag1992}.

In the presence of a constant magnetic field, to my knowledge, the author is the first one to consider the Hartree equation when $\Gamma_{0}$ is not of trace class or a Hilbert-Schmidt operator. Since the operator $h$ is now the Pauli operator other than the Laplace operator, the spectrum changes from a continuous one to a discrete one and having no eigenspaces turns into the case that  eigenspaces are of infinite dimension. Even though we mainly care about the case when $\Gamma_{0}$ is not of trace class, to complete the picture, when $\Gamma_{0}$ is of trace class and $w=\frac{1}{|x|}$, we establish a global well-posedness result at the energy level in the appendix.

The explicit form of Equation (\ref{hatree,density,form}) is 
\begin{equation}
    i\,\partial_{t}\Gamma=\left[-\partial_{x^3}^{2}+D^*D+b+\rho_{\Gamma}*w,\Gamma\right],
\end{equation}
where 
\begin{equation}
    D=-2\partial_{\bar{z}}-\frac{b}{2}z,\quad D^*=2\partial_{z}-\frac{b}{2}\bar{z},\quad z=x^1+ix^2.
\end{equation}
Consider first the two dimensional problem 
\begin{align}\label{hatree,density,form,main}
    \left\{\begin{array}{l}
        i\,\partial_{t}\gamma=\left[ H+\rho_{\gamma}*v,\gamma\right],\\
        \gamma(0,x,y)=\gamma_{0}(x,y),   
    \end{array}\right. \quad x,y\in\mathbb{R}^2,
\end{align}
where 
\begin{equation}
    H=D^*D,\quad \rho_{\gamma}(t,x)=\gamma(t,x,x),
\end{equation} 
and $\gamma: L^2\left(\mathbb{R}^2\right)\rightarrow L^2\left(\mathbb{R}^2\right)$. If $v\in L^{1}(\mathbb{R}^2)$ \footnote{For the given family of solutions $\bar{\Pi}_{\phi}$, $\rho_{\bar{\Pi}_{\phi}}=\bar{\Pi}_{\phi}(x,x)=\phi(0)$ is constant. In order for $\rho_{\bar{\Pi}_{\phi}}*v$ to make sense, $v\in L^{1}(\mathbb{R}^2)$.}, Equation (\ref{hatree,density,form,main}) admits one family \footnote{For the other family, see Section \ref{stationary,solutions}.} of non-trace class stationary solutions with integral kernels in the following form
\begin{equation}\label{stationary,solution,form,1}
    \bar{\Pi}_{\phi}(x,y)=\phi\left(|x-y|\right)e^{-i\frac{b\Omega(x,y)}{2}},\quad \Omega(x,y):=x^1y^2-x^2y^1,\ x,y\in\mathbb{R}^2,
\end{equation}
whose derivation is in Section \ref{stationary,solutions}. Inspired by \cite{LS2015,LS2014,CHP2017,CHP2018}, we are interested in the evolution of perturbations of the stationary solutions.

Suppose the pertubation of the stationary solution $\bar{\Pi}_{\phi}$ is $Q(t,x,y)=\gamma(t,x,y)-\bar{\Pi}_{\phi}(x,y)$, then the evolution equation for $Q$ is
\begin{align}\label{equation,pertubation}
    \left\{\begin{array}{l}
        \displaystyle i\,\partial_{t}Q =[H+\rho_{Q}*v,Q]+[\rho_{Q}*v,\bar{\Pi}_{\phi}]  \\
        \displaystyle Q(0,x,y)=Q_{0}(x,y),
    \end{array}\right.\quad x,y\in\mathbb{R}^2,
\end{align}
where $\rho_{Q}(t,x)=Q(t,x,x)$.

The operator $H$ has a discrete spectrum $\sigma(H)=\{2bj\}_{j\in\mathbb{N}}$ and it is decomposed into mutually orthogonal projections $P_{j}$ on $L^2(\mathbb{R}^2)$ with corresponding eigenvalue $2bj$,
\begin{equation*}
    H=\sum_{j=0}^{\infty}2bj\,P_{j}=\frac{b}{2\pi}\sum_{j=0}^{\infty}2bj L_{j}\left(\frac{b}{2}|x-y|^2\right)\exp\left(-\frac{b}{4}|x-y|^2\right)e^{-i\frac{b\Omega(x,y)}{2}}
\end{equation*}
where $P_{j}$ are infinite-dimensional projections and $L_{k}(\lambda)$ are Laguerre polynomials, i.e.
\begin{equation}
     L_{k}(\lambda)=\sum_{j=0}^{k}\binom{k}{j}\frac{(-\lambda)^j}{j!},\ \left(\lambda\in\mathbb{R}\right).
\end{equation}
For more details, see Section \ref{prop,H}.

The physical interpretation of $\bar{\Pi}_{\phi}$ is that when $\phi$ is chosen as 
\begin{equation}\label{stationary,landau}
    \phi\left(x\right)=\frac{b}{2\pi}\sum_{j=0}^{n}L_{j}\left(\frac{b}{2}|x|^2\right)\exp\left(-\frac{b}{4}|x|^2\right),\quad x\in\mathbb{R}^2,
\end{equation}
$\bar{\Pi}_{\phi}$  corresponds to the projection from $L^2\left(\mathbb{R}^2\right)$ onto the first $n+1$ eigenspaces \footnote{In the physics literature, they are called Landau levels.} of $H$, i.e. the possible low energy states of $H$. As an analog of the classical picture of a Fermi sea, we call $\bar{\Pi}_{\phi}$  the Fermi sea. The stationary solution associated to (\ref{stationary,landau}) covers an important physical example in our setting. Let $k_{B}$ be the Boltzmann's constant and $T$ be the absolute temperature, the Fermi-Dirac distribution in the operator form is given by
\begin{equation}\label{Fermi,Dirac,stat}
    \frac{1}{e^{(H-\mu)/k_{B}T}+1}f:=\sum_{j=0}^{\infty}\frac{1}{e^{(2bj-\mu)/k_{B}T}+1}P_{j}f,
\end{equation}
where $f\in L^2(\mathbb{R}^2)$. Setting $\mu=2nb$, the zero temperature limit ($T\rightarrow 0^{+}$) of (\ref{Fermi,Dirac,stat}) is $\mathbf{1}_{\left(H\le 2nb\right)}$, which is exactly the projection associated to (\ref{stationary,landau}). More generally, for any finite $\mu$, at zero or positive temperature, the Fermi-Dirac distribution corresponds to a $\bar{\Pi}_{\phi}$, where
\begin{equation}\label{cover,Fermi,Dirac}
    \phi\left(x\right)=\frac{b}{2\pi}\sum_{j=0}^{\infty}\frac{1}{e^{(2bj-\mu)/k_{B}T}+1}L_{j}\left(\frac{b}{2}|x|^2\right)\exp\left(-\frac{b}{4}|x|^2\right),\quad x\in\mathbb{R}^2.
\end{equation}

From a functional calculus perspective of the stationary solutions (\ref{stationary,solution,form,1}), suppose $l$ is a function defined on the spectrum $\sigma(H)$, i.e. $l$ determines a sequence, $l(H)$ is defined as 
\begin{equation}
    l(H):=\sum_{j=0}^{\infty}l\left(2bj\right)P_{j}=\frac{b}{2\pi}\sum_{j=0}^{\infty}l(2bj)L_{j}\left(\frac{b}{2}|x-y|^2\right)\exp\left(-\frac{b}{4}|x-y|^2\right)e^{-i\frac{b\Omega(x,y)}{2}},
\end{equation}
where we denote $l_{j}=l(2bj)$. $l(H)$ corresponds to $\phi$ in (\ref{stationary,solution,form,1}) in the way
\begin{equation}\label{phi,l,functional,correspondence}
    \phi(x)=\frac{b}{2\pi}\sum_{j=0}^{\infty}l_{j}L_{j}\left(\frac{b}{2}|x|^2\right)\exp\left(-\frac{b}{4}|x|^2\right).
\end{equation}

For our main results, we will use the following norms

\begin{definition}\label{def,main,norm} Suppose $f\in L^2\left(\mathbb{R}^2\right)$, $s\ge 0$,
\begin{align*}
    &\left\|H^{s/2}f\right\|_{L^2}^2:=\sum_{j=0}^{\infty} \left(2bj\right)^s\left\|P_{j}f\right\|_{L^2}^2,\quad \left\|\langle H\rangle^{s/2}f\right\|_{L^2}^2:=\sum_{j=0}^{\infty} \left\langle 2bj\right\rangle^s\left\|P_{j}f\right\|_{L^2}^2,\\
    &\left\|\bar{H}^{s/2}f\right\|_{L^2}:=\left\|H^{s/2}\bar{f}\right\|_{L^2},\quad \left\|\langle \bar{H}\rangle^{s/2}f\right\|_{L^2}:=\left\|\langle H\rangle^{s/2}\bar{f}\right\|_{L^2},
\end{align*}
where $\langle 2bj\rangle=\left(1+\left(2bj\right)^2\right)^{1/2}$ and $\bar{H}$ is the complex conjugation of $H$, i.e.
\[
    \bar{H}=\overline{D^*}\bar{D}=\left(2\partial_{\bar{z}}-\frac{b}{2}z\right)\left(-2\partial_{z}-\frac{b}{2}\bar{z}\right).
\]
\end{definition}

With respect to the new norms, we obtain a local well-posedness result of Equation (\ref{equation,pertubation}). To state the result, first recall that a mild solution of Equation (\ref{equation,pertubation}) is a solution satisfying the integral equation
\begin{equation}
    Q(t,x,y)=e^{-it(H_{x}-\bar{H}_{y})}Q_{0}(x,y)-i\int_{0}^{t}e^{-it(H_{x}-\bar{H}_{y})(t-\tau)}\left[\rho_{Q}*v,Q+\bar{\Pi}_{\phi}\right]\, d\tau
\end{equation}
in an appropriate space. The appropriate space in this paper is defined as a Banach space $\mathbf{N}_{T}$ endowed with the norm,
\begin{equation}
    \|Q(t,x,y)\|_{\mathbf{N}_{T}}:= \sup_{(q,r)\in S}\|\langle H_{x}\rangle^{1/2} \langle \bar{H}_{y}\rangle^{1/2}Q(t,x,y)\|_{\left(L_{I_{T}}^{q} L_x^{r} L_{y}^2\right)\cap \left(L_{I_{T}}^{q}L_{y}^{r}L_{x}^2\right)} +\| \langle\nabla_{x}\rangle^{9/8}\rho_{Q}(t,x)\|_{L_{I_{T}}^{2}L^2_{x}},
\end{equation}
where $I_{T}=[0,T]$ and
\begin{equation}\label{admissble,partial}
    S=\left\{(q,r)\left| \left(\frac{1}{q},\frac{1}{r}\right)\hbox{ is in the line segment connecting } \left(\frac{1}{\infty},\frac{1}{2}\right)\hbox{ and } \left(\frac{1}{4},\frac{1}{4}\right)\right.\right\}.
\end{equation}
The first part of $\mathbf{N}_{T}$ is the Strichartz norm and the set (\ref{admissble,partial}) is a subset of admissible pairs $(q,r)$ which satisfy 
\begin{equation}\label{admissible,pairs}
    \frac{1}{q}+\frac{1}{r}=\frac{1}{2},\quad 2<q\le\infty.
\end{equation}
The second part of $\mathbf{N}_{T}$ involves the collapsing term $\rho_{Q}$, whose estimate is the main new ingredient in this paper. The theorem that we want to prove is as follows
\begin{theorem}\label{pauli,local,wellposed}
Consider Equation (\ref{equation,pertubation}) and suppose that $v\in L^1\left(\mathbb{R}^2\right)$ and
\begin{equation}\label{condition,stationary}
    \phi(x)=\phi(|x|), \quad\left\|\langle H\rangle^{1/2}\langle \bar{H}\rangle^{1/2}\phi\right\|_{L^2}<\infty,\quad x\in\mathbb{R}^2.
\end{equation}
If the initial data $Q_{0}(x,y)$ satisfies
\[
    \left\Vert\langle H_{x}\rangle^{1/2} \langle \bar{H}_{y}\rangle^{1/2}Q_{0}(x,y)\right\Vert_{L^2_x L^2_y}<\infty,
\]
then for sufficiently short time $T$,  Equation (\ref{equation,pertubation}) has a mild solution in the Banach space $\mathbf{N}_{T}$.
\end{theorem} 
\begin{remark}
$\left\|\langle H\rangle^{1/2}\langle \bar{H}\rangle^{1/2}\phi\right\|_{L^2}$ is essentially $\left\|D\bar{D}\phi\right\|_{L^2}+\left\|D\phi\right\|_{L^2}+\left\|\bar{D}\phi\right\|_{L^2}+\left\|\phi\right\|_{L^2}$. By the relation (\ref{phi,l,functional,correspondence}), the condition on the corresponding $\{l_{j}\}$ is $\sum_{j=0}^{\infty}j^2l^2_{j}<\infty$. Thus (\ref{cover,Fermi,Dirac}) satisfies the condition (\ref{condition,stationary}).
\end{remark}
\begin{remark}
For the Banach space $\mathbf{N}_{T}$, we can increase the size of the set $S$ as long as it does not include to endpoint $\left(\frac{1}{2},\frac{1}{\infty}\right)$. Consequently, the existence time may decrease. 
\end{remark}

Since the norm $\mathbf{N}_{T}$ contains $\| \langle\nabla_{x}\rangle^{9/8}\rho_{Q}(t,x)\|_{L_{I_{T}}^{2}L^2_{x}}$, the proof of Theorem \ref{pauli,local,wellposed} is based on the following collapsing estimate.

\begin{theorem}[Collapsing Estimate]\label{pauli,collapsing,estimate}
Suppose $\gamma(t,x,y)=e^{-i(H_{x}-\bar{H}_{y})t}\gamma_{0}(x,y)$ is the solution to the linear equation
\begin{align}\label{homogeneous,pauli}
    \left\{\begin{array}{l}
        i\,\partial_{t}\gamma=\left[H,\gamma\right]  \\
        \gamma(0,x,y)=\gamma_{0}(x,y)\in L^2_{x}L^2_{y},
    \end{array}\right.\quad x, y\in\mathbb{R}^2,
\end{align}
the collapsing term $\rho_{\gamma}(t,x)=\gamma(t,x,x)$ satisfies
\begin{equation}
   \left\|\rho_{\gamma}(t,x)\right\|_{L^2_{[0,\pi/b]}L^2_{x}}\lesssim_{b}\left\|\langle H_{x}\rangle^{s/2}\langle H_{y}\rangle^{s/2}\gamma_{0}(x,y)\right\|_{L^2_{x}L^2_{y}},\quad s>\frac{1}{2}, 
\end{equation}
and
\begin{equation}\label{pauli,collapsing,estimate,case1}
    \left\||\nabla_{x}|^c \rho_{\gamma}(t,x)\right\|_{L^2_{[0,\pi/b]}L^2_{x}}\lesssim_{c,b}\left\|\langle H_x\rangle^{1/2}\langle \bar{H}_{y}\rangle^{1/2}\gamma_{0}(x,y)\right\|_{L^2_{x}L^2_{y}},\quad 0\le c<\frac{5}{4}.
\end{equation}
\end{theorem}
\begin{remark}
The estimate (\ref{pauli,collapsing,estimate,case1}) is only stated for the time interval $[0,\pi/b]$. However, since the solution $\gamma(t,x,y)$ has a period $\pi/b$, by a patching argument, (\ref{pauli,collapsing,estimate,case1}) holds for arbitrary large time interval $[-T,T]$, while the constant will depend on $T$.
\end{remark}

This type of estimates has been established in \cite{GM17,CH2016,CHP2017} for the Laplacian case, i.e. $i\,\partial_{t}\gamma=[-\Delta,\gamma]$. However the technique used in those papers does not apply to the current case. That method, in the spirit of \cite{Klainerman2008}, is to study the characteristic hypersurface, which is derived by applying the space-time Fourier transform after we collapse the solution $e^{it(\Delta_{x}-\Delta_{y})}\gamma_{0}$ to the diagonal $y=x$.  In our case, the time Fourier transform is replaced by the Fourier series. The new ingredients are the Fourier-Wigner transform and a refined estimate about the asymptotic property of associated Laguerre polynomials. 

The paper is organized in the following way: in Section \ref{notation} we define most notations used in the paper; in Section \ref{prop,H} we discuss the propagator $e^{-iHt}$ and the spectral structure of $H$; in Section \ref{many,particle,H} we establish the collapsing estimate Theorem \ref{pauli,collapsing,estimate}; in Section \ref{main,sec} we first give a low regularity result for Equation (\ref{hatree,density,form,main}) to show that the ``forcing'' term $[\rho_{Q}*v,\bar{\Pi}_{\phi}]$ in Equation (\ref{equation,pertubation}) is a challenging term to handle and then prove Theorem \ref{pauli,local,wellposed}; in Section \ref{conclusion}, we pose open problems for future study. In the appendix, in Section \ref{Heisenberg,group}, we give a short review of the Heisenberg group; in Section \ref{stationary,solutions} we present two families of stationary solutions to Equation (\ref{hatree,density,form,main}); in Section \ref{global,trace,class}, we show the global well-posedness of Equation (\ref{hatree,density,form}) for the case when $\Gamma_{0}$ is of trace class and $w(x)=\frac{1}{|x|}$.

\section{Notation}\label{notation}
For the reader's convenience, we define most notations used in the paper in this section.

Let $\Omega$ denote the canonical symplectic form on $\mathbb{R}^2$,
\begin{equation}
    \Omega(x,y):=x^1y^2-x^2y^1,\quad x,y\in\mathbb{R}^2,
\end{equation}
and $I,J$ be matrices
\begin{equation}
    I:=\begin{pmatrix}1 & 0\\ 0 & 1\end{pmatrix},\quad J:=\begin{pmatrix} 0 & 1\\ -1 & 0\end{pmatrix}.
\end{equation}
Let $\mathcal{S}(\mathbb{R}^d)$ denote the Schwartz space on $\mathbb{R}^{d}$ and 
\begin{equation}
    \langle f,g\rangle:=\int_{\mathbb{R}^{d}}f(x)\bar{g}(x)\,dx.
\end{equation}

Let $a$ and $a^{\dagger}$ be annihilation and creation operators
\begin{equation}
    a:=\frac{x+b\partial_{x}}{\sqrt{2}},\qquad a^{\dagger}:=\frac{x-b\partial_{x}}{\sqrt{2}},\quad x\in\mathbb{R}.
\end{equation}
Denote normalized Hermite polynomials by $h_{j}$, $j\in\mathbb{N}$,
\begin{equation}
    h_{j}(x):=\frac{\left(a^{\dagger}\right)^{j}}{(b\pi)^{1/4}\sqrt{j!\,b^{j}}}e^{-\frac{x^2}{2b}},\quad  x\in\mathbb{R}.
\end{equation}
They satisfy $\langle h_{j},h_{k}\rangle=\delta_{jk}$. $H_{h}$ denotes the Hermite operator
\begin{equation}
    H_{h}:=-\Delta_{x}+\frac{b^2|x|^2}{4},\quad x\in\mathbb{R}^2.
\end{equation}

If $A\lesssim B$, there is a constant $C$ such that $A\le CB$. Furthermore $A\lesssim_{p,q} B$ means that the constant $C$ depends on parameters $p$ and $q$. If $A\sim B$, there are constants $C_1$ and $C_2$ such that $C_2>C_1>0$ and $C_1 A\le B \le C_2 A$. Furthermore $A\sim_{p,q} B$ means that $C_1$ and $C_2$ depend on $p$ and $q$. Let $\mathcal{D}(H)$ denote the domain of the operator $H$.

We use the following tools from the harmonic analysis in the phase space \cite{Folland89}. On the Hilbert space $L^2\left(\mathbb{R}\right)$, the Heisenberg representation $\beta$ is defined as
\begin{equation}
    \beta(p,q,t)f:=e^{i(p\hat{P}+q\hat{X}+tb)}f=e^{iqx+\frac{ibpq}{2}+ibt}f(x+pb),\quad f\in L^2\left(\mathbb{R}\right),\ x,p,q,t\in\mathbb{R},
\end{equation}
where $\hat{P}=-ib\partial_{x}$ and $\hat{X}$ denotes the multiplication by $x$. For simplicity, denote $\beta(p,q,0)f$ as $\beta(p,q)f$. Notice that $\beta$ is a unitary representation.

The twisted convolution between two functions $f,g$ is
\begin{equation}\label{twisted,convolution,def}
    \left(f\natural g\right)(x):=\int_{\mathbb{R}^2}f(x-y)g(y)e^{\frac{ib}{2}\Omega(x,y)}\,dy,\quad x\in\mathbb{R}^2,
\end{equation}
and the ``complex conjugate'' $\bar{\natural}$ is defined as
\begin{equation}
    \left(f\bar{\natural} g\right)(x):=\int_{\mathbb{R}^2}f(x-y)g(y)e^{-\frac{ib}{2}\Omega(x,y)}\,dy,\quad x\in\mathbb{R}^2.
\end{equation}
The Fourier-Wigner transform $V$ is defined as the matrix coefficient of the Heisenberg representation 
\begin{equation}
    V(f,g)(p,q):=\left\langle \beta(p,q)f,g\right\rangle=\int_{\mathbb{R}}e^{iqx+\frac{ibpq}{2}}f(x+pb)\bar{g}(x)\,dx \quad p,q\in\mathbb{R},
\end{equation}
and the Wigner transform $W$ is the Fourier transform of $V$
\begin{equation}
    W(f,g)(\xi,x):=\frac{1}{2\pi}\int_{\mathbb{R}^2}V(f,g)(p,q)e^{-i\xi p-ixq}\,dpdq \quad \xi,x\in\mathbb{R}.
\end{equation}
\begin{remark}
All these concepts can be defined similarly in higher dimensions.
\end{remark}

\section{Properties of $H$}\label{prop,H}
In this section, we discuss the one parameter unitary subgroup $e^{-iHt}$ generated by $-iH$, where 
\begin{equation}
    H=D^*D=-\partial_{x^1}^2-\partial_{x^2}^2-ib\,(x^2\partial_{x^1}-x^1\partial_{x^2})+\frac{b^2}{4}(|x^1|^2+|x^2|^2)-b,\quad b>0,
\end{equation}
and the spectral structure of $H$. The formula for $e^{-iHt}$ is derived by applying the metaplectic representation and it is given below.
\begin{theorem}\label{pauli,fundamental solution}
Given the Schr\"odinger equation
\begin{equation}
    i\,\partial_{t}f(t,x)=Hf(t,x),\quad f(0,x)=f_{0}(x)\in \mathcal{S}(\mathbb{R}^2),
\end{equation}
the formula for the solution is
\begin{equation}\label{constant,sol}
    \left(e^{-iHt}f_{0}\right)(x)=\left\{
    \begin{array}{lc}
    \displaystyle \frac{b\,e^{ibt}}{4\pi i\sin(bt)}\int_{\mathbb{R}^2}\exp\left(\frac{ib(x-y)^2}{4\tan(bt)}-\frac{ib}{2}\Omega(x,y)\right)f_{0}(y)\,dy, & t\neq \frac{\pi}{b}k\\
      f_{0}(x),  & t=\frac{\pi}{b}k
    \end{array}\right.
\end{equation}
where $k\in\mathbb{Z}$. 
\end{theorem}
\begin{proof}
Consider the metaplectic representation $\mu$ \cite[Chapter 4]{Folland89} from the metaplectic group $Mp(4,\mathbb{R})$ to the unitary group $U\left(L^2(\mathbb{R}^2)\right)$ of $L^2\left(\mathbb{R}^2\right)$, where the corresponding infinitesimal representation is
\begin{align}
    d\mu: &\ \mathfrak{sp}(4,\mathbb{R})\rightarrow \mathfrak{u}\left(L^2(\mathbb{R}^2)\right) \\
        &\mathcal{A}=\begin{pmatrix}A & B\\ C& -A^{T}\end{pmatrix} \mapsto -\frac{1}{2i}\begin{pmatrix}\hat{Q} & \hat{P}\end{pmatrix}\begin{pmatrix}
    A & B\\
    C & -A^{T}
    \end{pmatrix}\begin{pmatrix}0 & I\\ -I & 0 \end{pmatrix}
    \begin{pmatrix}\hat{Q} \\ \hat{P}\end{pmatrix},
\end{align}
where $\hat{Q}=\begin{pmatrix} x^1\\ x^2\end{pmatrix},\ \hat{P}=\begin{pmatrix} -i\partial_{x^1}\\ -i\partial_{x^2}\end{pmatrix}$ and $A^{T}$ denotes the transpose matrix of $A$. Under $d\mu$, $-i\left(H+b\right)\in \mathfrak{u}\left(L^2(\mathbb{R}^2)\right)$ corresponds to 
\[
    \mathcal{A}=\begin{pmatrix}
        bJ & \frac{b^2}{2}\,I\\
        -2\,I & bJ
    \end{pmatrix}\in \mathfrak{sp}(4,\mathbb{R}).
\]

In order to apply Theorem \ref{metaplectic,expression} from the appendix to get the integral Formula (\ref{constant,sol}), we need to compute the explicit form for the one parameter subgroup $\exp(\mathcal{A}t)$ in the symplectic group $Sp(4,\mathbb{R})$. Since $\mathcal{A}$ can be written as a sum of two commuting matrices
\[ 
    \begin{pmatrix}J & 0\\0 & J\end{pmatrix}\ \hbox{and}
    \begin{pmatrix}0 & \frac{b^2}{2}\,I\\-2\,I & 0\end{pmatrix},
\]
then
\begin{align*}
    \exp(\mathcal{A}t) &= \exp\left(\begin{pmatrix}J & 0\\0 & J\end{pmatrix}bt\right) \cdot \exp\begin{pmatrix}0 & \frac{b^2t}{2}\,I\\-2t\,I & 0\end{pmatrix}\\         &=\begin{pmatrix}\exp\left(Jbt\right) & 0 \\ 0 & \exp\left(Jbt\right)\end{pmatrix} \cdot\begin{pmatrix}\cos(bt)\,I & \frac{b}{2}\sin(bt)\,I \\-\frac{2}{b}\sin(bt)\,I & \cos(bt)\,I\end{pmatrix}.
\end{align*}
Applying Theorem \ref{metaplectic,expression}, we get for $f_0\in\mathcal{S}(\mathbb{R}^2)$,
\begin{equation}\label{fundamental,solution,phase,space}
    \left(\mu(\exp(\mathcal{A}t))f_0\right)(x)=\frac{1}{2\pi\cos(bt)}\int_{\mathbb{R}^2}\exp\left(-iS(t,x,\xi)\right)\hat{f}_0(-\xi)\,d\xi,
\end{equation}
where the phase function $S$ is
\[
    S(t,x,\xi)=\frac{\tan(bt)}{b}\left|\xi\right|^2 +x\xi+\tan(bt)\Omega(x,\xi)+\frac{b\tan(bt)}{4}|x|^2,\quad \xi=(\xi^1,\xi^2),\ x=(x^1,x^2).
\]

To obtain Formula (\ref{constant,sol}),
\begin{align}
    \left(\mu(\exp(\mathcal{A}t))f_0\right)(x)
        &=\frac{1}{(2\pi)^2\cos(bt)}\int_{\mathbb{R}^2}f_0(y)\,dy\int_{\mathbb{R}^2}\exp\left(-iS(x,\xi)+iy\xi\right)\,d\xi \nonumber\\
    &=\frac{1}{(2\pi)^2\cos(bt)}\int_{\mathbb{R}^2}f_0(y)\,dy\int_{\mathbb{R}^2}\exp\left(\frac{bi}{4\tan(bt)}\left(x-\tan(bt)Jx-y\right)^2\right)\nonumber\\
    &\quad \cdot \exp\left(-i\left[\frac{b\tan(bt)}{4}|x|^2+\frac{\tan(bt)}{b}\left(\xi+\frac{b}{2\tan(bt)}(x-\tan(bt)Jx-y)\right)^2\right]\right)\,d\xi\nonumber\\
    &=\frac{b}{4\pi i\sin(bt)}\int_{\mathbb{R}^2}\exp\left(\frac{ib}{4\tan(bt)}(x-y)^2-\frac{ib}{2}\Omega(x,y)\right)f_0(y)\,dy.\label{constant,sol,temp}
\end{align}
Let us denote (\ref{constant,sol,temp}) by $sol(t)f_{0}$.

Theorem \ref{metaplectic,expression} is valid as long as the matrix $\cos(bt)I$ is not degenerate. Since $\cos(bt)$ vanishes at $\displaystyle\frac{\pi}{2b}$, we only obtain the Formula (\ref{constant,sol}) for $\displaystyle t\in\left[0,\frac{\pi}{2b}\right)$. Next we show that Formula (\ref{constant,sol,temp}) is valid on $\mathbb{R}$. Formula (\ref{constant,sol,temp}) is defined when $t\in (0,\pi/b)$. By direct computation, 
\[
    sol(t+s)f_{0}=sol(t)sol(s)f_{0},\quad \text{for } t>0,\ s>0,\ t+s<\frac{\pi}{b},
\]
i.e. $sol(t)$ is a semigroup when $t\in(0,\pi/b)$. Besides $sol(t)$ is also continuous with respect to the strong operator topology when $t\in[0,\pi/b)$. This is because when $t\in[0,\pi/2b)$, we obtain Formula (\ref{constant,sol,temp}) by the metaplectic representation; when $t\in[\pi/2b,\pi/b)$, $sol(t)=sol(\pi/2b)sol(t-\pi/2b)$. Therefore, by the uniqueness of the one parameter unitary subgroup generated by $d\mu (\mathcal{A})$, $e^{-i(H+b)t}=sol(t)$ is true for $t\in[0,\pi/b)$. As $t\rightarrow \pi/b$, from (\ref{fundamental,solution,phase,space}), we see that the phase function $S(t,x,\xi)\rightarrow x\xi$ and $\mu\left(\exp(\mathcal{A}t)\right)f_{0}\rightarrow -f_{0}$ pointwise. By the dominant convergence theorem, $sol(t)f_{0}$ also converges to $-f_{0}$ in $L^2(\mathbb{R}^2)$. In summary, we have obtained Formula (\ref{constant,sol}) for $t\in[0,\pi/b]$ and showed that $e^{-iHt}$ is of period $\pi/b$. Therefore $e^{-i(H+b)t}=sol(t)$ holds for $t\in\mathbb{R}$.
\end{proof}
\begin{remark}
According to the metaplectic representation $\mu$, one can also conclude that $e^{-i(H+b)\pi/b}=-1$ by the observation that  $exp(\mathcal{A}t):[0,\pi/b]\rightarrow Sp(4,\mathbb{R})$ is the generator of the fundamental group $\pi_{1}\left(Sp(4,\mathbb{R})\right)$ of $Sp(4,\mathbb{R})$ and the metapletic group is the double cover of $Sp(4,\mathbb{R})$.
\end{remark}

Based on the formula (\ref{constant,sol}) and the machinery in \cite{GinibreVelo1992}, we obtain the Strichartz estimate to arbitrary finite time.

\begin{corollary}\label{strichartz,scalar,long,time}
Fix any time $T>0$, 
\begin{equation}\label{long,time,strichartz}
    \left\|e^{-iH t}f\right\|_{L^q_{[0,T]}L^r_{x}}\lesssim_{q,r,T}\left\|f\right\|_{L^2_{x}},
\end{equation}
where $(q,r)$ satisfies (\ref{admissible,pairs}).
\end{corollary}
\begin{proof}
For any $T> 0$, there is a positive integer $n$ such that $T_{\epsilon}=\frac{T}{n}\le \frac{\pi}{10b}$. Based on Theorem \ref{pauli,fundamental solution}, for $t<T_{\epsilon}$, 
\[
    \left\|e^{-i\,H t}f\right\|_{L^{\infty}(\mathbb{R}^2)}\lesssim\frac{1}{t}\left\|f\right\|_{L^1(\mathbb{R}^2)}.
\]
Since $e^{-iHt}$ is unitary, by \cite{GinibreVelo1992}, $\left\|e^{-i\,Ht}f\right\|_{L^q_{[0,T_{\epsilon}]}L^r_{x}}\lesssim_{q,r} \|f\|_{L^2_{x}}$, where $(q,r)$ satisfies (\ref{admissible,pairs}). For any integer $j$, $1\le j\le n$, repeat the above argument on the time interval $[(j-1)T_{\epsilon},jT_{\epsilon}]$, 
\[
    \left\|e^{-iHt}f\right\|_{L^q_{[(j-1)T_{\epsilon},jT_{\epsilon}]}L^2_x}\lesssim_{q,r}\left\|e^{-iH(j-1)T_{\epsilon}}f\right\|_{L^2_x}=\|f\|_{L^2}.
\]
Then apply the Minkowski inequality
\[
    \left\|e^{-i\,Ht}f\right\|_{L^q_{[0,nT_{\epsilon}]}L^r_{x}}\le \sum_{j=1}^{n}\left\|e^{-i\,Ht}f\right\|_{L^q_{[(j-1)T_{\epsilon},jT_{\epsilon}]}L^r_{x}}\lesssim_{q,r} n\|f\|_{L^2_x}.
\]
\end{proof}

The spectrum of $H$ is well-known in the physics literature. Here we give a discussion of its spectral structure and some formulas based on the Fourier-Wigner transform. $H$ is a non-negative self-adjoint operator on $L^2(\mathbb{R}^2)$. Since for any $f\in\mathcal{D}(D)$,
\[
    Df=\left(-2\partial_{\bar{z}}-\frac{b}{2}z\right)f=e^{-b|z|^2/4}\left(-2\partial_{\bar{z}}\right)\left(e^{b|z|^2/4}f\right),
\]
and $\partial_{\bar{z}}$ is elliptic, the null space $\mathcal{H}_{0}$ of $H$ consists of all functions in the form $g(z)e^{-b|z|^2/4}$, where $g(z)$ is an entire function. To rephrase it, $e^{b|z|^2/4}\mathcal{H}_{0}$ is a Fock-Bargmann space\cite[Section 1.6]{Folland89} with probability measure $b\,e^{-b|z|^2/2}d\mu/2\pi$, where $d\mu$ is the Lebesgue measure on $\mathbb{C}$. Thus, with respect to the canonical Hermitian inner product on $L^2(\mathbb{R}^2)$, $\mathcal{H}_{0}$ has an orthonormal basis 
\begin{equation}
    e_{0j}(z):=\frac{z^{j}}{\sqrt{\pi j!(2/b)^{j+1}}}\exp\left(-\frac{b|z|^2}{4}\right),\quad z\in\mathbb{C}\simeq\mathbb{R}^2,\ j\in\mathbb{N},
\end{equation}
and the integral kernel $P_{0}(x,y)$ associated to the projection $P_{0}:L^2(\mathbb{R}^2)\rightarrow \mathcal{H}_{0}$ is
\begin{align}
    P_{0}(x,y)&=\frac{b}{2\pi}\exp\left(-\frac{b|x-y|^2}{4}-\frac{ib}{2}\,\Omega(x,y)\right)\label{null,space,kernel}\\
        &=\frac{b}{2\pi}\exp\left(-\frac{b}{4}\left(|z_{x}|^2+|z_{y}|^2\right)+\frac{b}{2}z_{x}\bar{z}_{y}\right),\nonumber
\end{align}
where $z_{x}=x^1+ix^2$, $x=(x^1,x^2)\in\mathbb{R}^2$, $z_{y}=y^1+iy^2$ and $y=(y^1,y^2)\in\mathbb{R}^2$.

Using the commutation relation $[D,(D^*)^k]=2bk(D^*)^{k-1}$, we obtain other eigenspaces $\mathcal{H}_{k}=(D^*)^k(\mathcal{H}_{0})$ associated to eigenvalue $2bk$ and orthonormal bases of $\mathcal{H}_{k}$ for $k\in\mathbb{N}$,
\begin{equation}
    e_{kj}(z):=\frac{(D^*)^k}{\sqrt{(2b)^k k!}}e_{0j}(z),\quad j\in\mathbb{N}.
\end{equation}
On the level of eigenspaces, $H$ has a ladder operator structure $D^*\mathcal{H}_{k}=\mathcal{H}_{k+1}$ and $D(\mathcal{H}_k)=\mathcal{H}_{k-1}$. Therefore we call $D$ and $D^*$ annihilation and creation operators respectively. Furthermore,
\begin{lemma}\label{L2,pauli,completeness}
The space $L^2(\mathbb{R}^2)$ is decomposed orthogonally as follows
\[
    L^2(\mathbb{R}^2)=\overline{\bigoplus_{k\in\mathbb{N}}\mathcal{H}_{k}},
\]
which implies that $H$ has a discrete spectrum $\sigma(H)=\left\{ 2bk\right\}_{k=0}^{\infty}$ with corresponding eigenspaces $\mathcal{H}_{k}$.
\end{lemma}
\begin{proof}
Consider the related Hermite operator $H_{h}$, $x=(x^1,x^2)\in\mathbb{R}^2$,  and associated creation and annihilation operators
\[
    a^{\dagger}_{j}=\partial_{x^{j}}-\frac{b}{2}x^{j},\ a_{j}=-\partial_{x^{j}}-\frac{b}{2}x^{j},\quad j\in\{1,2\}.
\]
$\left\{\left(a_{1}^{\dagger}\right)^{j}\left(a_2^{\dagger}\right)^{l}e^{-b|x|^2/4}\right\}_{j,l\in\mathbb{N}}$ is a basis for $L^2(\mathbb{R}^2)$. Since 
\begin{align*}
    \left(a_{1}^{\dagger}\right)^{j}\left(a_2^{\dagger}\right)^{l}e^{-b|x|^2/4} &=e^{b|x|^2/4}(\partial_{x^1})^j(\partial_{x^2})^l e^{-b|x|^2/2}\\
        &=i^l\, e^{b|x|^2/4}(\partial_{z}+\partial_{\bar{z}})^j(\partial_{z}-\partial_{\bar{z}})^l e^{-b|x|^2/2},
\end{align*}
and bases of $\mathcal{H}_{k}$ are in the form 
\[
    (D^*)^{k}\left(z^j e^{-b|z|^2/4}\right)=e^{b|z|^2/4}(2\partial_{z})^{k}\left(-\frac{2}{b}\partial_{\bar{z}}\right)^{j}e^{-b|z|^2/2},
\]
$\left(a_{1}^{\dagger}\right)^{j}\left(a_2^{\dagger}\right)^{l}e^{-b|x|^2/4}$ can be written as a linear combination of  bases of $\mathcal{H}_{k}$. Therefore the $L^2$-closure of $\bigoplus_{k\in\mathbb{N}}\mathcal{H}_{k}$ is $L^2(\mathbb{R}^2)$.
\end{proof}

Another perspective to derive the spectrum of $H$ is first decomposing it as a sum of three operators: the constant operator $-b$, the Hermite operator $H_{h}$ and the rotation vector field $H_{r}=-ib(x^2\partial_{x^1}-x^1\partial_{x^2})=\bar{z}\partial_{\bar{z}}-z\partial_{z}$, i.e.
\begin{equation}\label{pauli,operator,decompose}
    H=H_{h}+H_{r}-b.
\end{equation}
The three operators all commute with each other. Thus they all share same eigenvectors. More precisely, 
\[
    H_{h}e_{kj}=(k+j+1)b\,e_{kj},\quad H_{r}e_{kj}=(k-j)b\,e_{kj}.
\]
Then $He_{kj}=\left(H_{h}+\bar{z}\partial_{\bar{z}}-z\partial_{z}-b\right)e_{kj}=2kb\,e_{kj}$. Displaying all eigenvectors $e_{kj}$ schematically in Figure \ref{fig}, all rows correspond to different eigenspaces of $H$, all columns correspond to different eigenspaces of the complex conjugate $\bar{H}$, all lines with slope equal to $-1$ correspond to different eigenspaces of $H_{h}$ and all lines slope equal to $1$ correspond to different eigenspaces of $H_{r}$.

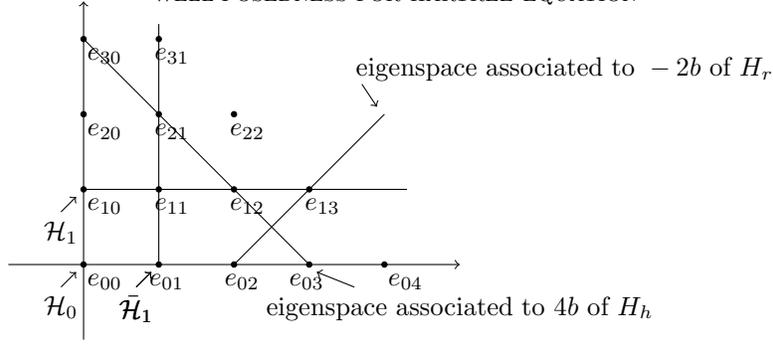
\begin{figure}[hbt]
    \centering
    \begin{tikzpicture}
    \draw[black,->] (-1,0) -- (5,0); 
    \draw[black,->] (0,-1) -- (0,3.5);
    \filldraw[black] (0,0) circle (1pt) node[anchor=north, xshift=8pt] {$e_{00}$};
    \filldraw[black] (1,0) circle (1pt) node[anchor=north, xshift=3pt] {$e_{01}$};
    \filldraw[black] (2,0) circle (1pt) node[anchor=north, xshift=3pt] {$e_{02}$};
    \filldraw[black] (3,0) circle (1pt) node[anchor=north, xshift=-1pt] {$e_{03}$};
    \filldraw[black] (4,0) circle (1pt) node[anchor=north, xshift=8pt] {$e_{04}$};
    \filldraw[black] (0,1) circle (1pt) node[anchor=north, xshift=8pt] {$e_{10}$};
    \filldraw[black] (0,2) circle (1pt) node[anchor=north, xshift=8pt] {$e_{20}$};
    \filldraw[black] (0,3) circle (1pt) node[anchor=north, xshift=8pt] {$e_{30}$};
    \filldraw[black] (1,1) circle (1pt) node[anchor=north, xshift=5pt] {$e_{11}$};
    \filldraw[black] (2,1) circle (1pt) node[anchor=north, xshift=5pt] {$e_{12}$};
    \filldraw[black] (3,1) circle (1pt) node[anchor=north, xshift=5pt] {$e_{13}$};
    \filldraw[black] (1,2) circle (1pt) node[anchor=north, xshift=5pt] {$e_{21}$};
    \filldraw[black] (2,2) circle (1pt) node[anchor=north, xshift=5pt] {$e_{22}$};
    \filldraw[black] (1,3) circle (1pt) node[anchor=north, xshift=5pt] {$e_{31}$};
    \draw[black,->] (-0.3,-0.3) -- (-0.1,-0.1);
    \node[anchor=north] at (-0.3,-0.3) {$\mathcal{H}_{0}$};
    \draw[black,->] (-0.3,0.7) -- (-0.1,0.9);
    \node[anchor=north] at (-0.3,0.7) {$\mathcal{H}_{1}$};
    \draw[black,->] (0.7,-0.3) -- (0.9,-0.1);
    \node[anchor=north] at (0.7,-0.3) {$\bar{\mathcal{H}}_{1}$};
    \draw[black,->] (0.7,-0.3) -- (0.9,-0.1);
    \node[anchor=north] at (0.7,-0.3) {$\bar{\mathcal{H}}_{1}$};
    \draw[black,->] (3.6,-0.3) -- (3.1,-0.1);
    \node[anchor=north] at (5,-0.3) {$\text{eigenspace associated to } 4b\text{ of }H_{h}$};
    \draw[black,->] (3.7,2.4) -- (3.9,2.1);
    \node[anchor=west] at (3.5,2.6) {$\text{eigenspace associated to } -2b\text{ of }H_{r}$};
    \draw (0,1) -- (4.3,1);
    \draw (1,0) -- (1,3.2);
    \draw (3,0) -- (0,3);
    \draw (2,0) -- (4,2);
    \end{tikzpicture}
    \caption{Canonical Eigenfunctions of $H$}
    \label{fig}
    {$e_{kj}$ are in the form:
    $e_{kj}(z)=\left\{\begin{array}{cc}
        z^{j-k}p\left(\frac{b|z|^2}{2}\right) & j\ge k, \\
        \bar{z}^{k-j}p\left(\frac{b|z|^2}{2}\right) & k>j,
    \end{array}\right.$
    where $p$ is a polynomial of degree $\min\{j,k\}$.
    }
\end{figure}

For the null space $\mathcal{H}_0$ of $H$, since there is a reproducing kernel in the Fock-Bargmann space, $P_{0}(x,y)$ is given by (\ref{null,space,kernel}). Based on this formula (\ref{null,space,kernel}) and the ladder structure $\mathcal{H}_{k}=\left(D^*\right)^{k}\mathcal{H}_{0}$, we obtain the following expressions for all projections $P_{k}$ in terms of the Fourier-Wigner transform. 

\begin{lemma}\label{H,spectral,structure}
 The projection $P_{k}$ associated to the eigenspace $\mathcal{H}_{k}$ can be expressed as
\begin{equation}
    P_{k}f=\frac{b}{2\pi}V(h_{k},h_{k})\bar{\natural}f,\quad f\in L^2\left(\mathbb{R}^2\right).
\end{equation}
More explicitly, the integral kernel $P_{k}(x,y)$ of $P_{k}$ is 
\begin{equation}\label{Pauli,operator,projection kernel}
    P_{k}(x,y)=\frac{b}{2\pi}L_{k}\left(\frac{b|x-y|^2}{2}\right)\exp\left(-\frac{b|x-y|^2}{4}-\frac{ib}{2}\,\Omega(x,y)\right),\quad x,y\in\mathbb{R}^2,
\end{equation}
where $\displaystyle L_{k}(\lambda)=\sum_{j=0}^{k}\binom{k}{j}\frac{(-\lambda)^j}{j!},\ (\lambda\in\mathbb{R})$ are Laguerre polynomials.
\end{lemma}
\begin{proof}
Suppose $g\in \mathcal{S}(\mathbb{R})$, the Fourier-Wigner transform of $g$ and $e^{-\lambda^2/2b}$ is
\begin{align*}
    V\left(g,e^{-\lambda^2/2b}\right)(x^1,x^2)&=\int_{\mathbb{R}}e^{ix^2\lambda+ibx^1x^2/2}g(\lambda+bx^1)e^{-\lambda^2/2b}\,d\lambda\\
    &= e^{-b|z|^2/4}\int_{\mathbb{R}}e^{\lambda z-\lambda^2/2b-bz^2/4}g(\lambda)\,d\lambda,
\end{align*} 
where $z=x^1+i\,x^2$. $\left(g\mapsto\int_{\mathbb{R}}e^{\lambda z-\lambda^2/2b-bz^2/4}g(\lambda)\,d\lambda\right)$ defines a Bargmann transform from $L^2(\mathbb{R})$ to the Fock-Bargmann space with weight $e^{-b|z|^2/2}d\mu$. Since the correspondence is isomorphic, we identify $L^2(\mathbb{R})$ with $\mathcal{H}_{0}$. Notice that $D^*$ and the creation operator $a^{\dagger}$ are connected through the identity
\begin{equation}
    \frac{D^{*}}{\sqrt{2}}V(g,e^{-\lambda^2/2b})=V\left(g,a^{\dagger}e^{-\lambda^2/2b}\right).
\end{equation}
Then $L^{2}(\mathbb{R})$ corresponds to $\mathcal{H}_{k}=(D^*)^k\mathcal{H}_{0}$ by $\left(g\mapsto V\left(g,\left(a^{\dagger}\right)^k e^{-\lambda^2/2b}\right)\right)$. Therefore for any $f\in L^{2}(\mathbb{R}^2)$, there is a sequence $\{f_{k}\}_{k\in\mathbb{N}}\subset L^{2}(\mathbb{R})$ such that 
\[
    f(x)=\sum_{k=0}^{\infty}V\left(f_{k},\frac{\left(a^{\dagger}\right)^k}{(\pi b)^{1/4}\sqrt{k!(2b)^k}} e^{-\lambda^2/2b}\right)=\sum_{k=0}^{\infty}V\left(f_{k},h_{k}\right).
\]
By Lemma \ref{V,W,properties} and Theorem \ref{V,Hermit,Laguerre} from the appendix, 
\begin{align*}
    &\quad\overline{V\left(h_{j},h_{j}\right)}\natural \overline{ f(x)}=\sum_{k=0}^{\infty}\frac{2\pi}{b}\left\langle h_{k},h_{j}\right\rangle \overline{V\left(f_{k},h_{j}\right)}\\
    &=\frac{2\pi}{b}\sum_{k=0}^{\infty}\delta_{jk}\overline{V\left(f_{k},h_{j}\right)}=\frac{2\pi}{b}\overline{V(f_{j},h_{j})}\\
    \implies & \left\{\begin{array}{l}
        \displaystyle P_{k}=\frac{b}{2\pi} V\left(h_{k},h_{k}\right)\overline{\natural } \quad \text{or}\\
        \displaystyle P_{k}(x,y)=\frac{b}{2\pi}L_{k}\left(\frac{b}{2}|x-y|^2\right)\exp\left(-\frac{b|x-y|^2}{4}-\frac{ib\,\Omega(x,y)}{2}\right).
    \end{array}
    \right. 
\end{align*}
\end{proof}
\begin{remark} Similarly for $\bar{H}$, the projection $\bar{P}_{k}$ onto the $k$-th eigenspace of $\bar{H}$ is
\begin{equation}
   \bar{P}_{k}f=\frac{b}{2\pi}V(h_{k},h_{k})\natural f,\quad  f\in L^2(\mathbb{R}^2),
\end{equation}
and the integral kernel of $\bar{P}_{k}$ is simply the complex conjugation of $P_{k}(x,y)$.
\end{remark}

\begin{remark}
    $D$ commutes with complex conjugates $\bar{D}$ and $\bar{D}^*$.
\end{remark}

At the end of this section, we list some results about $H$ for later use. 

The difference between $H^{1/2}$ and $D$ can be analogous to the one between $(-\Delta)^{1/2}$ and $\nabla$. Generally  for any $f\in\mathcal{D}\left(H^{1/2}\right)$, $H^{1/2}f$ is not the same as $Df$. The difference is apparent when decompose $f$ as $f=\sum_{k=0}^{\infty}P_{k}f$ and apply $H^{1/2}$ and $D$ to $f$ separately
\[
    H^{1/2}f=\sum_{k=1}^{\infty}(2bk)^{1/2}P_{k}f,\quad Df=\sum_{k=1}^{\infty}DP_{k}f,
\]
where $(2bk)^{1/2}P_{k}f$ and $DP_{k}f$ are in $\mathcal{H}_{k}$ and $\mathcal{H}_{k-1}$ respectively. However, they have the same $L^2$ norms
\begin{equation}\label{L2,equivalent}
    \left\langle H^{1/2}f,H^{1/2}f\right\rangle=\left\langle Hf,f\right\rangle=\left\langle D^*Df, f\right\rangle=\left\langle Df,Df\right\rangle.
\end{equation}
More generally, for $1<p<\infty$, 
\begin{equation}\label{Lp,H,D,relation}
    \|(H+b)^{1/2}f\|_{L^{p}(\mathbb{R}^2)}\sim_{p}\|Df\|_{L^{p}(\mathbb{R}^2)}+\left\|D^*f\right\|_{L^p(\mathbb{R}^2)}.
\end{equation} 
\begin{remark}
    To see why (\ref{Lp,H,D,relation}) is true, note that our vector field potential $A$ satisfies $A\in L^{2}_{loc}(\mathbb{R}^3)^3$ and the magnetic field $\mathbf{B}=(0,0,b)$ is constant. Then by \cite[Theorem 1.3, 1.6]{Ali2010}, for $1<p<\infty$,
    \[
        \left\|Lf\right\|_{L^p(\mathbb{R}^3)}\sim_{p}\left\|\left(H-\partial^2_{x^3}+b\right)^{1/2}f\right\|_{L^p(\mathbb{R}^3)},
    \]
    where $L=\left(-i\partial_{x^1}+\frac{b}{2}x^2,-i\partial_{x^2}-\frac{b}{2}x^1,-i\partial_{x^3}\right)$ and $x=(x^1,x^2,x^3)$. Since in the third dimension it is known that $\|\partial_{x^3}g\|_{L^p(\mathbb{R})}\sim_{p}\|(-\partial^2_{x^3})^{1/2}g\|_{L^p(\mathbb{R})}$ and 
    \[
        \left\|\left(-i\partial_{x^1}+\frac{b}{2}x^2,-i\partial_{x^2}-\frac{b}{2}x^1\right)f\right\|_{L^p(\mathbb{R}^2)}\sim \left\|Df\right\|_{L^p(\mathbb{R}^2)}+\left\|D^*f\right\|_{L^p(\mathbb{R}^2)},
    \]
    we obtain (\ref{Lp,H,D,relation}).
\end{remark} 
Unlike $(-\Delta)^{1/2}$ and $\nabla$, where they both commute with $\Delta$, $\left[D,H\right]=2bD$.

There is no comparison between $\left\|\nabla f\right\|_{L^2}$ and $\left\|D f\right\|_{L^2}$. For example, $\left\|D e_{0k}\right\|_{L^2_{\mathbb{R}^2}}=0$, for any $k\in\mathbb{N}$. While
\[
    \left\|\nabla e_{0k}\right\|_{L^2_{\mathbb{R}^2}} =\left\|2\partial_{\bar{z}} e_{0k}\right\|_{L^2_{\mathbb{R}^2}}=\frac{\sqrt{(k+1)b}}{\sqrt{2}}\left\|e_{0k+1}\right\|_{L^2_{\mathbb{R}^2}}=\frac{\sqrt{(k+1)b}}{\sqrt{2}}
\]
blows up as $k$ approaches infinity. On the other hand, taking $f\in C_{c}^{\infty}(\mathbb{R}^2)$, consider the translation $f_{\tilde{x}}=f(x-\tilde{x})$, then
\[
    \left\|\nabla f_{\tilde{x}}\right\|_{L^2_{\mathbb{R}^2}}=\left\|\nabla f\right\|_{L^2_{\mathbb{R}^2}},\quad \left\|D f_{\tilde{x}}\right\|_{L^2_{\mathbb{R}^2}}\rightarrow \infty\ \hbox{as}\  \tilde{x}\rightarrow \infty.
\]
However there is a pointwise identity, for $f,g\in \mathcal{S}(\mathbb{R}^2)$,
\[
    -2\partial_{\bar{z}}(f\bar{g})=(Df)\bar{g}-f\overline{D^*g},
\]
which implies 
\begin{equation}\label{pointwise,covariant,comparison}
    \left|\partial_{\bar{z}}\left|f\right|\right|=\left|\frac{\partial_{\bar{z}}\left(f\bar{f}\right)}{2|f|}\right|\le \frac{\left|\left(Df\right)\bar{f}\right|}{|2f|}+\frac{\left|f\overline{D^* f}\right|}{2|f|}=\frac{1}{2}\left(|Df|+\left|D^*f\right|\right),
\end{equation}
i.e. $\left|\partial_{\bar{z}}\left|f\right|\right|\lesssim |Df|+\left|D^*f\right|$. Based on this observation, we can still make use of the Sobolev inequality.

\begin{lemma}\label{covariant,n-endpoint,sobolev}
For $2\le q<\infty$,
\begin{equation} 
    \|f\|_{L^q(\mathbb{R}^2)}\lesssim_{q}\left\|\langle H\rangle^{1/2} f\right\|_{L^2(\mathbb{R}^2)}.
\end{equation}
\end{lemma}
\begin{proof}
By (\ref{pointwise,covariant,comparison}),
\[
    \left\|\nabla |f|\right\|_{L^2}=\left\|-2\partial_{\bar{z}}|f|\right\|_{L^2}\le \|Df\|_{L^2}+\|D^*f\|_{L^2}\lesssim \|Df\|_{L^2}+\|f\|_{L^2},
\]
and apply the usual $n$-endpoint Sobolev inequality,
\begin{align*}
    \|f\|_{L^q(\mathbb{R}^2)}=\left\| |f| \right\|_{L^q(\mathbb{R}^2)}\lesssim_{q} \|f\|_{L^2(\mathbb{R}^2)}+\|\nabla |f|\|_{L^2(\mathbb{R}^2)}\lesssim \|f\|_{L^2(\mathbb{R}^2)}+\|D f\|_{L^2(\mathbb{R}^2)}.
\end{align*}
\end{proof}

\section{Strichartz and Collapsing Estimates}\label{many,particle,H}
In this section, we study the linear equation $i\,\partial_{t}\gamma=[H,\gamma]$. The formula of the propagator $e^{-iHt}$ and the spectral structure of $H$ from Section \ref{prop,H} are the basic tools for our discussion. Similar to Corollary \ref{strichartz,scalar,long,time}, for any finite time $T$, we obtain the Strichartz estimate for $e^{-iHt}\gamma_{0}e^{iHt}=e^{-i\left(H_{x}-\bar{H}_{y}\right)t}\gamma_{0}$.

\begin{proposition}\label{Strichartz,homogeneous}
Let $\gamma(t,x,y)=e^{-i(H_{x}-\bar{H}_{y})t}\gamma_{0}(x,y)$ be the solution to Equation (\ref{homogeneous,pauli}),
then for any $T>0$ and $s\ge0$,
\begin{equation}\label{long,time,homogeneous,matrix}
    \left\|\langle H_{x}\rangle^{s/2} \langle \bar{H}_{y}\rangle^{s/2}\gamma(t,x,y)\right\|_{\left(L^{q}_{[0,T]}L^{r}_{x}L^2_{y}\right)\cap \left(L^{q}_{[0,T]}L^{r}_{y}L^2_{x}\right) }\lesssim_{q,r,T} \left\|\langle H_{x}\rangle^{s/2} \langle \bar{H}_{y}\rangle^{s/2}\gamma_0(x,y)\right\|_{L^2_{x}L^2_{y}},
\end{equation}
where $(q,r)$ satisfies (\ref{admissible,pairs}). Furthermore, by duality, the following dual estimate holds
\begin{equation}\label{dual,strichartz}
    \left\|\int_{0}^{T}e^{i(\bar{H}_{x}-H_{y})t}F(t,x,y)\,dt\right\|_{L^2_{x}L^2_{y}}\lesssim_{q',r',T}\left\|F(t,x,y)\right\|_{\left(L^{q'}_{[0,T]}L^{r'}_{x}L^2_{y}\right)\cap \left(L^{q'}_{[0,T]}L^{r'}_{y}L^2_{x}\right)},
\end{equation}
where 
\[
    \frac{1}{q}+\frac{1}{q'}=1,\quad \frac{1}{r}+\frac{1}{r'}=1.
\]
\end{proposition}
\begin{proof}
The two statements in (\ref{long,time,homogeneous,matrix}) are symmetric with respect to $x$ and $y$, we show the estimate for one of them and the other one is obtained by swapping roles of $x$ and $y$. Apply $\langle H_{x}\rangle^{s/2} \langle \bar{H}_{y}\rangle^{s/2}$ to Equation (\ref{homogeneous,pauli}),
\[
    i\,\partial_{t}\langle H_{x}\rangle^{s/2} \langle \bar{H}_{y}\rangle^{s/2}\gamma(t,x,y)=\left(H_{x}-\bar{H}_{y}\right)\langle H_{x}\rangle^{s/2} \langle \bar{H}_{y}\rangle^{s/2}\gamma(t,x,y).
\]
View $e^{-i\left(H_{x}-\bar{H}_{y}\right)t}$ as a map on the Hilbert space of $L^2(\mathbb{R}^2)$-valued functions. It is unitary since the Hilbert space  $\left\{f\left| f:\mathbb{R}^2\rightarrow L^{2}\left(\mathbb{R}^2\right)\right.\right\}$ is canonically isometric to $L^2\left(\mathbb{R}^2\times\mathbb{R}^2\right)$.
Besides, using Formula (\ref{constant,sol}), for $t<T_{\epsilon}\le\frac{\pi}{10b}$, 
\begin{align*}
    \left\|e^{-i\left(H_{x}-\bar{H}_{y}\right)t}\langle H_{x}\rangle^{s/2} \langle \bar{H}_{y}\rangle^{s/2}\gamma_{0}\right\|_{L^{\infty}_{x}L^{2}_{y}}&\lesssim \frac{1}{t}\left\|e^{i\bar{H}_{y}t}\langle H_{x}\rangle^{s/2} \langle \bar{H}_{y}\rangle^{s/2}\gamma_{0}\right\|_{L^1_{x}L^2_{y}}\\
        &=\frac{1}{t}\left\|\langle H_{x}\rangle^{s/2} \langle \bar{H}_{y}\rangle^{s/2}\gamma_{0}\right\|_{L^1_{x}L^2_{y}}.
\end{align*}
Then by the abstract version of the Strichartz estimate \cite[Theorem 10.1]{KeelTaoEn}, 
\begin{align*}
     \left\|\langle H_{x}\rangle^{s/2} \langle \bar{H}_{y}\rangle^{s/2}\gamma(t,x,y)\right\|_{L^{q}_{[0,T_{\epsilon}]}L^{r}_{x}L^2_{y}}\lesssim_{q,r} \left\|\langle H_{x}\rangle^{s/2} \langle \bar{H}_{y}\rangle^{s/2}\gamma_0(x,y)\right\|_{L^2_{x}L^2_{y}}.
\end{align*}
Following the same patching argument as Corollary \ref{strichartz,scalar,long,time}, we obtain the estimate (\ref{long,time,homogeneous,matrix}).
\end{proof}

In order to  show Theorem \ref{pauli,collapsing,estimate},
we need to decompose the initial data $\gamma_{0}(x,y)$ based on the spectral structures of $H_{x}$ and $\bar{H}_{y}$. According to Lemma \ref{H,spectral,structure},
\begin{align}
    \gamma_{jk}(x,y) &:= P_{xj}\bar{P}_{yk}\gamma_{0}=V(h_{j},h_{j})\overline{\natural}_{x}V(h_{k},h_{k})\natural_{y}\gamma_{0}\\
    \gamma_{0}(x,y) &=\sum_{j,k\in\mathbb{N}}\gamma_{jk}(x,y)\nonumber\\
        &=\sum_{j,k\in\mathbb{N}}\int_{\mathbb{R}^{2}\times\mathbb{R}^2}V(h_{j},h_{j})(x-\tilde{x})V(h_{k},h_{k})(y-\tilde{y})e^{-ib\left[\Omega(x,\tilde{x})-\Omega(y,\tilde{y})\right]/2}\gamma_{jk}(\tilde{x},\tilde{y})\, d\tilde{x} d\tilde{y},\label{decompostion,initial}
\end{align}
where $P_{xj}(\bar{P}_{yk})$ means the projection of $\gamma_{0}(x,y)$ onto $\mathcal{H}_{j}(\bar{\mathcal{H}}_{k})$ with respect to the $x(y)$ variable. Then in the kernel form, the evolution of $\gamma_{0}$ under Equation (\ref{homogeneous,pauli}) can be expressed as 
\begin{equation}
    \left(e^{-(H_{x}-\bar{H}_{y})it}\gamma_{0}\right)(x,y) =\sum_{j,k\in\mathbb{N}}e^{-2b(j-k)it}\gamma_{jk}(x,y).
\end{equation}

In the later computation of the space Fourier transform of (\ref{decompostion,initial}), associated Laguerre polynomials $L^{\alpha}_{n}(\lambda)$ appear in the collapsing term 
\begin{equation}
    L^{\alpha}_{n}(\lambda)=\sum_{j=0}^{n}\binom{n+\alpha}{n-j}\frac{(-\lambda)^{j}}{j!}, \quad \lambda\in \mathbb{R},\ n\in\mathbb{N},\ \alpha>-1.
\end{equation}
Thus estimates about these polynomials are needed for the collapsing estimate and we discuss them first.

\begin{lemma}\label{Laguerre,asymptotic,c0}
For $j,n,c\in\mathbb{N}$,
\begin{equation}\label{Laguerre,asymptotic,easy1}
    \frac{n!}{\left(n+j\right)!}\max_{\lambda\ge 0}\lambda^{j+c}\left(L^j_{n}\right)^2(\lambda)e^{-\lambda}\le 4^c(j+2n+c)^{c}.
\end{equation}
Furthermore, since associated Laguerre polynomials are related to $V\left(h_{j},h_{k}\right)$ by Theorem \ref{V,Hermit,Laguerre} in Appendix \ref{transform}, (\ref{Laguerre,asymptotic,easy1}) is equivalent to
\begin{equation}\label{Laguerre,asymptotic,easy}
    \left\|\bar{w}^{c}V\left(h_{j},h_{k}\right)(p,q)\right\|_{L^{\infty}}\le \left(\frac{2}{\sqrt{b}}\right)^{c}\left(j+k+c\right)^{c/2},
\end{equation}
where $j,k,c\in\mathbb{N}$, and $w=p+iq\in\mathbb{C}$.
\end{lemma}
\begin{proof}
We prove (\ref{Laguerre,asymptotic,easy}) by induction on $c$. First consider the basic case $c=0$, by Cauchy-Schwartz inequality, for $j,k\in\mathbb{N}$,
\[
    \left| V\left(h_{j},h_{k}\right)\right|=\left|\langle\beta(p,q)h_{j},h_{k}\rangle\right|\le \left\|\beta(p,q)h_{j}\right\|_{L^2}\left\|h_{k}\right\|_{L^2}=\left\|h_{j}\right\|_{L^2}\left\|h_{k}\right\|_{L^2}=1.
\]
Assume (\ref{Laguerre,asymptotic,easy}) holds for $c=n\in\mathbb{N}$. When $c=n+1$, using the following commutation relations,
\[
    \left[a,a^{\dagger}\right]=b,\quad \left[a,\beta(p,q)\right]=-\frac{b}{\sqrt{2}}\bar{w}\beta(p,q),
\]
we obtain
\begin{align*}
    \bar{w}^{n+1}V\left(h_{j},h_{k}\right)(p,q)&=-\frac{\sqrt{2}}{b}\bar{w}^{n}\left\langle [a,\beta(p,q)]h_{j},h_{k}\right\rangle\\
        &=-\frac{\sqrt{2}}{b}\left(\sqrt{(k+1)b}\,\bar{w}^{n}V\left(h_{j},h_{k+1}\right)-\sqrt{jb}\,\bar{w}^{n}V\left(h_{j-1},h_{k}\right)\right)(p,q).
\end{align*}
Using the induction assumption, $\left|\bar{w}^{n+1}V\left(h_{j},h_{k}\right)(p,q)\right|$
\begin{align*}
    &\le \frac{2^{n+1/2}}{b^{(n+1)/2}}\left(\sqrt{k+1}(j+k+1+n)^{n/2}+\sqrt{j}(j+k+n-1)^{n/2}\right)\\
    &\le \left(\frac{2}{\sqrt{b}}\right)^{n+1}(j+k+n+1)^{(n+1)/2}.
\end{align*}
Therefore (\ref{Laguerre,asymptotic,easy}) holds for all $c\in\mathbb{N}$.
\end{proof}

There is a more refined estimate than Lemma \ref{Laguerre,asymptotic,c0},
\begin{theorem}\label{Ilia,result}\cite{Krasikov05,Krasikov07}
Let $n\ge 1$, $\alpha>-1$, then
\[
    \frac{n!}{\mathcal{G}(n+\alpha+1)}\max_{\lambda\ge 0}\left(\lambda^{\alpha+1}e^{-\lambda}\left(L_{n}^{\alpha}\right)^2(\lambda)\right)<6n^{1/6}\sqrt{n+\alpha+1},
\]
where $\mathcal{G}$ denotes the gamma function
\[
    \mathcal{G}(z):=\int_{0}^{\infty}\lambda^{z-1}e^{-\lambda}\,d\lambda,\quad \mathcal{R}(z)>0.
\]
\end{theorem}
Notice that Krasikov's result is for the case $c=1$ in Lemma \ref{Laguerre,asymptotic,c0}. In the case $c=1$, the upper bound in  (\ref{Laguerre,asymptotic,easy1}) is essentially $(j+n)$ for large $n$ and $j$. When considering the asymptotic behavior of  $\frac{n!}{\left(n+j\right)!}\max_{\lambda\ge 0}\lambda^{j+1}\left(L^j_{n}\right)^2(\lambda)e^{-\lambda}$ in terms of $j$ and $n$, Krasikov's result is sharper. If we interpolate Krasikov's result with Lemma \ref{Laguerre,asymptotic,c0}, we improve (\ref{Laguerre,asymptotic,easy1}) a little bit.

\begin{lemma}\label{lemma,Laguerre,asympotic,c,not,good}
Let $1\le c\le 2$,
\begin{equation}\label{Laguerre,asympotic,c,not,good}
    \frac{n!}{(n+j)!}\max_{\lambda\ge 0}\left(\lambda^{j+c}e^{-\lambda}\left(L_{n}^{j}\right)^2(\lambda)\right)\lesssim (1+n)^{(2-c)/6}(n+j+1)^{(3c-2)/2},\quad  j,n\in\mathbb{N},
\end{equation}
or equivalently,
\begin{equation}
    \left\||w|^{c}V\left(h_{j},h_{k}\right)(p,q)\right\|_{L^{\infty}}\lesssim \frac{1}{b^{c/2}}(1+k)^{(2-c)/12}(j+1)^{(3c-2)/4},
\end{equation}
where $j,k\in\mathbb{N}$, $j\ge k$ and $w=p+iq\in\mathbb{C}$.
\end{lemma}
\begin{proof}
Two endpoint cases of (\ref{Laguerre,asympotic,c,not,good}) are $c=1$ and $c=2$. 

The case $c=2$ is given by taking $c=2$ in (\ref{Laguerre,asymptotic,easy1}).

The case $c=1$ is almost in Theorem \ref{Ilia,result} except for $n=0$. When $n=0$, by Stirling formula,
\[
    \frac{1}{(j)!}\max_{\lambda\ge 0}\left(\lambda^{j+1}e^{-\lambda}\left(L_{0}^{j}\right)^2(\lambda)\right)=\frac{(j+1)^{j+1}e^{-(j+1)}}{j!}\lesssim \sqrt{j+1}.
\]
Combining it with Theorem \ref{Ilia,result},
\[
    \frac{n!}{(n+j)!}\max_{\lambda\ge 0}\left(\lambda^{j+1}e^{-\lambda}\left(L_{n}^{j}\right)^2(\lambda)\right)\lesssim (1+n)^{1/6}\sqrt{n+j+1},\quad j,n\in\mathbb{N}.
\]

For any fixed $\lambda>0$, vary the exponent $\alpha$ in $\lambda^{j+1+\alpha}e^{-\lambda}\left(L_{n}^{j}\right)^2(\lambda)$, where $0\le \mathcal{R}(\alpha)\le 1$. Interpolating the two endpoint cases, (\ref{Laguerre,asympotic,c,not,good}) holds.
\end{proof}
\begin{remark}
Lemma \ref{lemma,Laguerre,asympotic,c,not,good} is stated for $1\le c\le2$. Because this is what we need in the present case. Nevertheless, using Krasikov's result, we can improve (\ref{Laguerre,asymptotic,easy1}) for any $c\ge 1$.
\end{remark}

\begin{remark}\label{associated,Laguerre,not,optimal}
The upper bound in Lemma \ref{lemma,Laguerre,asympotic,c,not,good} is not optimal. Consider two extreme cases $j=0$ and $n=0$ of 
\begin{equation}\label{asymptotic,Laguerre}
    \frac{\sqrt{n!}}{\sqrt{(n+j)!}}\max_{\lambda\ge 0}e^{-\lambda/2}\lambda^{(c+j)/2}\left|L_{n}^{j}\right|(\lambda)\sim\left\||w|^c V(h_{n},h_{n+j})(w)\right\|_{L^{\infty}},\quad c\ge\frac{1}{2}.
\end{equation}  

\cite[Theorem 8.91.2, p.~241]{Gabor75} says for any $a>0$ and any fixed $j\in\mathbb{N}$,
\begin{equation*}
    \sup_{\lambda\ge a} e^{-\lambda/2}|\lambda|^{(c+j)/2}\left|L_{n}^{j,c}\right|(\lambda) \sim_{j} \langle n\rangle^{j/2+c/2-1/3},\quad c\ge 1/2.  
\end{equation*}
Taking $j=0$, one can remove the constraint $\lambda\ge a>0$ and show that $\max_{\lambda\ge 0} e^{-\lambda/2}|\lambda|^{c/2}\left|L_{n}\right|(\lambda) \sim \langle n\rangle^{c/2-1/3},\ c\ge 1/2$. It gives a precise description of the asymptotic behavior of (\ref{asymptotic,Laguerre}) for case $j=0$.

For the case $n=0$ of (\ref{asymptotic,Laguerre}), by Stirling formula,
\[
    \frac{1}{\sqrt{j!}}\max_{\lambda\ge0}e^{-\lambda/2}\lambda^{(c+j)/2}\left|L^{j}_{0}\right|(\lambda)=\frac{e^{-(c+j)/2}(c+j)^{(c+j)/2}}{\sqrt{j!}}\lesssim_{c} \langle j\rangle^{c/2-1/4}.
\]

``Interpolating'' the two cases, we conjecture  
\begin{equation}\label{asymptotic,Lagurre,precise}
    \left\||w|^c V(h_{n},h_{n+j})(w)\right\|_{L^{\infty}}\lesssim_{c} \langle n\rangle^{-1/12}\langle n+j\rangle^{c/2-1/4},\quad c\ge 1/2,\ j,n\in\mathbb{N}.
\end{equation}
When $c=1$, $n\ge 50$ and $j\ge 11$, by \cite[Theorem 2]{KraZar2010},  (\ref{asymptotic,Lagurre,precise}) holds. For other cases, our numerical data, for example Figure \ref{numerical,evidence}, strongly suggests that  (\ref{asymptotic,Lagurre,precise}) might hold. 
\begin{figure}[htb]
    \centering
    \begin{subfigure}{.5\textwidth}
        \centering
        \includegraphics[width=0.9\linewidth]{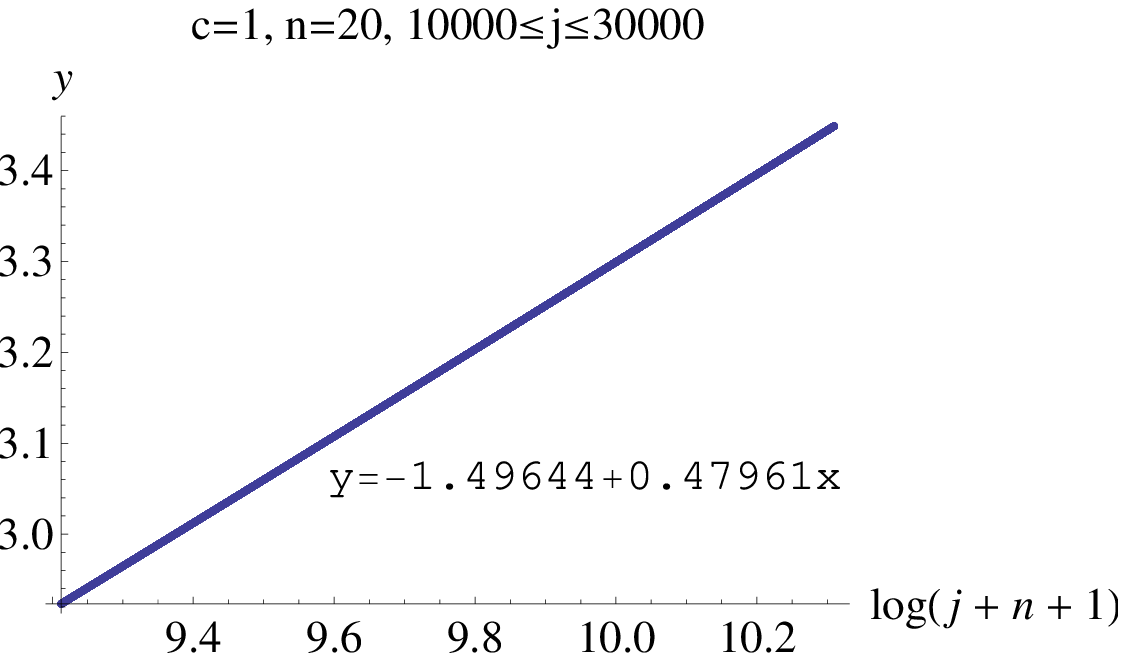}
        \caption{}
    \end{subfigure}%
    \begin{subfigure}{.5\textwidth}
        \centering
        \includegraphics[width=0.9\linewidth]{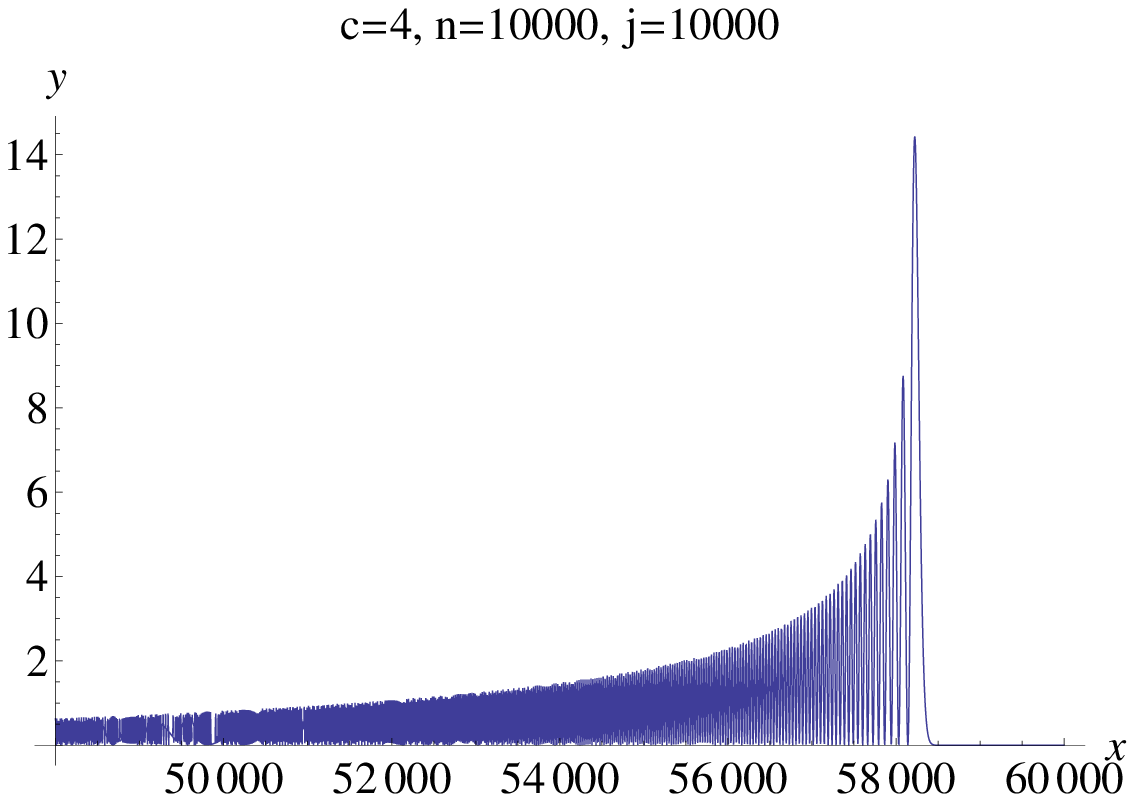}
        \caption{}
    \end{subfigure}
    \caption{Numerical Calculations}
    \medskip
    \small
    {In (A), to confirm the case $c=1$, $y=\log\left(\frac{n!}{(n+j)!}\max_{\lambda\ge0}e^{-\lambda}\lambda^{1+j}\left|L^{j}_{0}\right|^2(\lambda)\right)$ and $x=\log(j+n+1)$, where $n=20$, the data almost lies on a line. For a larger range of $j$, the slope is close to $0.5$. In (B), $y=\frac{n!(1+n)^{1/6}}{(n+j)!(n+j+1)^{c-1/2}}e^{-\lambda}\lambda^{(c+j)}\left|L_{n}^{j}\right|^2(\lambda)$. For fixed $c$, when we vary $n$ and $j$, from our numerical observation, $y$ is uniformly bounded. If we increase $c$, the bound increases.
    \label{numerical,evidence}}
\end{figure}
\end{remark}

Now we are ready to establish the collapsing estimate Theorem \ref{pauli,collapsing,estimate}.
\begin{proof}
By the Parseval's theorem on $L^2([0,\pi/b])$,
\begin{align}
    \left\||\nabla_{x}|^c \gamma(t,x,x)\right\|_{L^2_{[0,\pi/b]}L^2_{x}}^2&=\left\||\nabla_{x}|^c\sum_{j,k\in\mathbb{N}}e^{-2b(j-k)it}\gamma_{jk}(x,x)\right\|_{L^2_{[0,\pi/b]}L^2_{x}}^2 \nonumber\\
        &=\frac{\pi}{b}\sum_{m\in\mathbb{Z}}\left\||\nabla_x|^c\sum_{\substack{j-k=m\\ j,k\in\mathbb{N}}}\gamma_{jk}(x,x)\right\|_{L^2_x}^2. \label{collapsing,term,c}
\end{align}
We will express (\ref{collapsing,term,c}) by the Fourier transform of $|\nabla_{x}|^c\gamma_{jk}(x,x)$. Using the expression (\ref{decompostion,initial}),
\begin{align*}
    \left(\widehat{|\nabla_{x}|^c\gamma_{jk}(x,x)}\right)(\xi)&=\frac{1}{2\pi}\int_{\mathbb{R}^4}d\tilde{x}d\tilde{y}\,|\xi|^c\gamma_{jk}(\tilde{x},\tilde{y})\left(e^{-i\xi\tilde{x}}W(h_{j})(\xi)\right)*\left(e^{-i\xi\tilde{y}}W(h_{k})(\xi)\right)\\
        &\quad*\delta\left(\xi+\frac{b}{2}J(\tilde{x}-\tilde{y})\right),
\end{align*} 
where $W(h_{j})=W(h_{j},h_{j})$.
To compute $\left(e^{-i\xi\tilde{x}}W(h_{j})(\xi)\right)*\left(e^{-i\xi\tilde{y}}W(h_{k})(\xi)\right)$, using tools from Appendix \ref{transform}
\begin{align*}
    & \left(e^{-i\xi\tilde{x}}W(h_{j})(\xi)\right)*\left(e^{-i\xi\tilde{y}}W(h_{k})(\xi)\right)\\
    &=\int_{\mathbb{R}^2}d\tilde{\xi}\, e^{-i\left(\xi-\tilde{\xi}\right)\tilde{x}}W(h_j)\left(\xi-\tilde{\xi}\right)e^{-i\tilde{\xi}\tilde{y}}W(h_{k})\left(\tilde{\xi}\right)\\
    &=\int_{\mathbb{R}^2}d\tilde{\xi}\, e^{-i\left(\xi-\tilde{\xi}\right)\tilde{x}}W(h_j)\left(\tilde{\xi}-\xi\right)e^{-i\tilde{\xi}\tilde{y}}W(h_{k})\left(\tilde{\xi}\right)\\
    &=e^{-i\xi\tilde{x}/2}\int_{\mathbb{R}^2}d\tilde{\xi}\,W\left(\beta\left(\frac{\tilde{x}}{2}-\frac{J\xi}{b}\right)h_{j} ,\beta\left(-\frac{\tilde{x}}{2}-\frac{J\xi}{b}\right)h_j\right)(\tilde{\xi}) \overline{W\left(\beta\left(\frac{\tilde{y}}{2}\right)h_{k},\beta\left(-\frac{\tilde{y}}{2}\right)h_k\right)}(\tilde{\xi})\\
    &=\frac{2\pi e^{-i\xi\tilde{x}/2}}{b}\left\langle \beta\left(\frac{\tilde{x}}{2}-\frac{J\xi}{b}\right)h_{j}, \beta\left(\frac{\tilde{y}}{2}\right)h_{k} \right\rangle  \left\langle \beta\left(-\frac{\tilde{y}}{2}\right)h_k, \beta\left(-\frac{\tilde{x}}{2}-\frac{J\xi}{b}\right)h_j \right\rangle\\
    &=\frac{2\pi e^{-i\xi\tilde{x}/2}}{b}\left\langle \beta\left(-\frac{\tilde{y}}{2}\right)\beta\left(\frac{\tilde{x}}{2}-\frac{J\xi}{b}\right)h_{j}, h_{k} \right\rangle  \left\langle \beta\left(\frac{\tilde{x}}{2}+\frac{J\xi}{b}\right)\beta\left(-\frac{\tilde{y}}{2}\right)h_k, h_j \right\rangle\\
    &=\frac{2\pi e^{-i\left(\tilde{x}+\tilde{y}\right)\xi}}{b}V\left(h_{j},h_{k}\right)\left(\frac{\tilde{x}-\tilde{y}}{2}-\frac{J\xi}{b}\right)V\left(h_{k},h_{j}\right)\left(\frac{\tilde{x}-\tilde{y}}{2}+\frac{J\xi}{b}\right).
\end{align*}
Then
\begin{align*}
    \left(\widehat{|\nabla_{x}|^c\gamma_{jk}(x,x)}\right)(\xi)&=\frac{1}{b}\int_{\mathbb{R}^4}d\tilde{x}d\tilde{y}\,|\xi|^c\gamma_{jk}(\tilde{x},\tilde{y})\exp\left(-\frac{i}{2}\left[(\tilde{x}+\tilde{y})\xi+\frac{b}{2}\,\Omega(\tilde{x}+\tilde{y},\tilde{x}-\tilde{y})\right]\right)\times\\
    &\quad V(h_{j},h_{k})\left(\tilde{x}-\tilde{y}-\frac{J\xi}{b}\right)V(h_{k},h_{j})\left(\frac{J\xi}{b}\right).
\end{align*}

Next estimate (\ref{collapsing,term,c}), using the Fourier transform on $\tilde{x}+\tilde{y}$ and the Minkowski inequality,
\begin{align*}
    (\ref{collapsing,term,c})& \lesssim b^{-3}\sum_{m\in\mathbb{Z}}\left(\sum_{\substack{j-k=m\\ j,k\in\mathbb{N}}} \left\|\int_{\mathbb{R}^2}d(\tilde{x}-\tilde{y})\, |\xi|^c\mathcal{F}_{\tilde{x}+\tilde{y}}\left(\gamma_{jk}\right)\left(\frac{1}{2}\left(\xi+\frac{bJ(\tilde{x}-\tilde{y})}{2}\right),\tilde{x}-\tilde{y}\right)\times\right.\right.\\
        &\quad \left.\left.V(h_{j},h_{k})\left(\tilde{x}-\tilde{y}-\frac{J\xi}{b}\right)V(h_{k},h_{j})\left(\frac{J\xi}{b}\right)\right\|_{L^2_{\xi}}\right)^2\\
    &\lesssim b^{-3}\sum_{m\in\mathbb{Z}}\left(\sum_{\substack{j-k=m\\ j,k\in\mathbb{N}}} \left\|V(h_{j},h_{k})\right\|_{L^2_{\mathbb{R}^2}}\|\gamma_{jk}\|_{L^2_{\mathbb{R}^4}}\sup_{\xi\in\mathbb{R}^2}|\xi|^c\left|V(h_{k},h_{j})\right|\left(\frac{J\xi}{b}\right) \right)^2\\
    &\qquad (\hbox{Cauchy-Schwartz inequality})\\
    &\lesssim b^{-4}\sup_{m\in\mathbb{Z}}\left(\sum_{\substack{j-k=m\\ j,k\in\mathbb{N}}}\frac{1}{\langle 2bj\rangle^{s}\langle 2bk\rangle^s}\sup_{\xi\in\mathbb{R}^2}|\xi|^{2c}\left|V(h_{j},h_{k})\right|^2\left(\frac{J\xi}{b}\right)\right)\sum_{j,k\in\mathbb{N}}\langle 2bj\rangle^{s}\langle 2bk\rangle^{s}\|\gamma_{jk}\|_{L^2_{\mathbb{R}^4}}^2.\\
    &\qquad \left(\hbox{since }\left\|V(h_{j},h_{k})\right\|_{L^2_{\mathbb{R}^2}}^2=\frac{2\pi}{b}\langle h_{k},h_{k}\rangle\langle h_{j},h_{j}\rangle=\frac{2\pi}{b}\right)
\end{align*}
The estimate (\ref{pauli,collapsing,estimate,case1}) reduces to show
\[
    b^{-4}\sup_{m\in\mathbb{N}}\left(\sum_{\substack{j-k=m\\ j,k\in\mathbb{N}}}\frac{1}{\langle 2bj\rangle^{s}\langle 2bk\rangle^s}\sup_{\xi\in\mathbb{R}^2}|\xi|^{2c}\left|V(h_{j},h_{k})\right|^2\left(\frac{J\xi}{b}\right)\right)<\infty.
\]

Taking $c=0$, by Lemma \ref{Laguerre,asymptotic,c0}, for any $m\in\mathbb{N}$,
\begin{align*}
    \sum_{\substack{j-k=m\\ j,k\in\mathbb{N}}}\frac{1}{\langle 2bj\rangle^{s}\langle 2bk\rangle^s}\left\|V(h_{j},h_{k})\left(\frac{J\xi}{b}\right)\right\|_{L^{\infty}}^2\le \sum_{\substack{j-k=m\\ k\in\mathbb{N}}} \frac{1}{\langle 2bj\rangle^{s}\langle 2bk\rangle^{s}}\le \sum_{k\in\mathbb{N}}\frac{1}{\langle 2bk\rangle^{2s}},
\end{align*}
which is finite if $s>1/2$. Taking $ 1\le c\le 2$, by Lemma \ref{lemma,Laguerre,asympotic,c,not,good},
\begin{align*}
    \sum_{\substack{j-k=m\\ j,k\in\mathbb{N}}}\frac{1}{\langle 2bj\rangle^{s}\langle 2bk\rangle^s}\left\||\xi|^cV(h_{j},h_{k})\left(\frac{J\xi}{b}\right)\right\|_{L^{\infty}}^2 &\lesssim \sum_{\substack{j-k=m\\ k\in\mathbb{N}}} \frac{b^{c}\langle k\rangle^{(2-c)/6}\langle j\rangle^{(3c-2)/2}}{\langle 2bj\rangle^{s}\langle 2bk\rangle^{s}}\\
    &\lesssim  \frac{1}{b^{2s-c}}\sum_{k\in\mathbb{N}}\frac{1}{\langle k\rangle^{2s-4c/3+2/3}},
\end{align*}
which is finite if $2s-4c/3+2/3>1$. Setting $s=1$, we get $1 \le c< 5/4$.

Combining the low frequency case $c=0$ and the high frequency case $1 \le c< 5/4$ yields the estimate (\ref{pauli,collapsing,estimate,case1}).
\end{proof}

\section{Well-Posedness of the System}\label{main,sec}
Before showing the local well-posedness result Theorem \ref{pauli,local,wellposed}, we discuss Equation (\ref{hatree,density,form,main}) in a case other than Equation (\ref{equation,pertubation}) to demonstrate that $[\rho_{Q}*v,_{\phi}]$ in Equation $(\ref{equation,pertubation})$ is a trouble term. Equation (\ref{hatree,density,form,main}) is well-posed in several spaces. The possible low regularity for the initial data when we can obtain a local well-posedness result is
\begin{equation}\label{unperturbed,strong}
    \left\| H_{hx}^{1/8+\epsilon} H_{hy}^{1/8+\epsilon}\gamma_{0}(x,y)\right\|_{L^2_{x}L^2_{y}}<\infty,\quad x,y\in\mathbb{R}^2,\quad \hbox{for arbitrary}\ \epsilon>0,
\end{equation}
where the norm is 
\[
    \left\| H_{h}^{s/2}f\right\|_{L^{2}}=\left\||\nabla|^{s}f\right\|_{L^2}+\left\||x|^{s}f\right\|_{L^2},\quad s\ge0,\ f\in L^{2}(\mathbb{R}^2).
\]
For the initial data (\ref{unperturbed,strong}), we acquire the following result. 

\begin{theorem}\label{local,low,regularity,pauli}
Consider Equation (\ref{hatree,density,form,main}) and suppose the initial condition $\gamma_{0}$ satisfies (\ref{unperturbed,strong}). Then Equation (\ref{hatree,density,form,main}) has a mild solution for sufficiently short time $T$ in the Banach $\mathbf{N}_{HT}$, where the norm is defined as
\begin{equation}\label{low,regularity,space}
    \|\gamma\|_{\mathbf{N}_{HT}}:=\left\|H_{hx}^{1/8+\epsilon}H_{hy}^{1/8+\epsilon}\gamma(t,x,y)\right\|_{L^{\infty}_{I_{T}}L^2_{x}L^2_{y}}+\left\||\nabla|^{1/2+2\epsilon}\rho_{\gamma}(t,x)\right\|_{L^2_{I_{T}}L^2_{x}},
\end{equation}
where $I_{T}=[0,T]$ and the $\epsilon$ is the same in (\ref{unperturbed,strong}).
\end{theorem}
\begin{remark}
Notice that the initial condition only requires that $\gamma_{0}$ is a Hilbert-Schmidt operator. It is not necessarily of trace class.
\end{remark}

In order to use the technique in \cite[Section 5, Section 6]{GM17} \footnote{The case studied in \cite{GM17} is in three dimension. However we can modify the argument for our two dimensional problem Equation (\ref{hatree,density,form,main}). Some steps in \cite{GM17} need minor modification, yet the main idea is the same.} to prove Theorem \ref{local,low,regularity,pauli}, we need another version of the collapsing estimate
\begin{proposition}\label{collapsing,pauli,hermite}
Suppose $\gamma(t,x,y)=e^{-i(H_{x}-\bar{H}_{y})}\gamma_{0}(x,y)$ is the solution to the linear equation
\begin{align}\label{homogeneous,pauli,1}
    \left\{\begin{array}{l}
        i\,\partial_{t}\gamma=\left[H,\gamma\right],   \\
        \gamma(0,x,y)=\gamma_{0}(x,y)\in L^2_{x}L^2_{y},  
    \end{array}\right.\quad x, y\in\mathbb{R}^2,
\end{align}
the collapsing term $\rho_{\gamma}(t,x)=\gamma(t,x,x)$ satisfies
\begin{equation}\label{pauli,lens,estimate}
    \left\|\langle\tan bt\rangle^{-1/2-\epsilon}|\nabla_{x}|^{1/2+2\epsilon} \rho_{\gamma}(t,x)\right\|_{L^2_{[-\pi/2b,\pi/2b]}L^2_{x}}\lesssim_{\epsilon}\left\|\langle \nabla_x\rangle^{1/4+\epsilon}\langle \nabla_{y}\rangle^{1/4+\epsilon}\gamma_{0}(x,y)\right\|_{L^2_{x}L^2_{y}},
\end{equation}
where $\epsilon$ is any arbitrary small positive number.
\end{proposition}
\begin{proof}
The operator $H$ is decomposed as (\ref{pauli,operator,decompose}) and $[H_{h},H_{r}]=0$. Since the rotation generated by the vector field $-iH_{r}$ satisfies 
\begin{equation*}
    \left||\nabla|^s e^{-iH_{r}t}f\right|(x)=\left||\nabla|^s f\right|\left(e^{-iH_{r}t}x\right),\quad x\in\mathbb{R}^2
\end{equation*}
and $e^{-i(H_{rx}-\bar{H}_{ry})t}\gamma_{0}(x,y)=e^{-i(H_{rx}+H_{ry})t}\gamma_{0}(x,y)$,
\begin{align*}
    &\left\|\langle\tan bt\rangle^{-1/2-\epsilon}|\nabla_{x}|^{1/2+2\epsilon} \rho_{\gamma}(t,x)\right\|_{L^2_{[-\pi/2b,\pi/2b]}L^2_{x}}\\
    =&\left\|\langle\tan bt\rangle^{-1/2-\epsilon}|\nabla_{x}|^{1/2+2\epsilon} \left(e^{-i\left(H_{rx}-\bar{H}_{ry}\right)t}e^{-i(H_{hx}-\bar{H}_{hy})t}\gamma_{0}\right)(x,x)\right\|_{L^2_{[-\pi/2b,\pi/2b]}L^2_{x}}\\
    =&\left\|\langle\tan bt\rangle^{-1/2-\epsilon}|\nabla_{x}|^{1/2+2\epsilon} \left(e^{-i(H_{hx}-\bar{H}_{hy})t}\gamma_{0}\right)(x,x)\right\|_{L^2_{[-\pi/2b,\pi/2b]}L^2_{x}}.
\end{align*}
Then the estimate (\ref{pauli,lens,estimate}) reduces to 
\begin{equation*}
    \left\|\langle\tan bt\rangle^{-1/2-\epsilon}|\nabla_{x}|^{1/2+2\epsilon} \left(e^{-i(H_{hx}-\bar{H}_{hy})t}\gamma_{0}\right)(x,x)\right\|_{L^2_{[-\pi/2b,\pi/2b]}L^2_{x}}\lesssim_{\epsilon}\left\|\langle \nabla_x\rangle^{1/4+\epsilon}\langle \nabla_{y}\rangle^{1/4+\epsilon}\gamma_{0}(x,y)\right\|_{L^2_{x}L^2_{y}},
\end{equation*}
where the collapsing term corresponds to the equation $i\,\partial_{t}\gamma=\left[H_{h},\gamma\right]$. By the Lens transform \cite{Tao09}
\begin{equation}
    \L(u)(t,x):=\frac{1}{\cos bt}u\left(\frac{\tan bt}{b},\frac{x}{\cos bt}\right)e^{-\left(ib|x|^2\tan bt\right)/4},\quad t\in\mathbb{R},\ x\in\mathbb{R}^2,
\end{equation}
which maps the solution $u(t,x)$ of $i\,\partial_{t}u=-\Delta u$ to the solution of $i\,\partial_{t}\L(u)=H_{h}\L(u)$, we obtain the identity
\begin{align*}
     &\left\|\langle\tan bt\rangle^{-1/2-\epsilon}|\nabla_{x}|^{1/2+2\epsilon} \left(e^{-i(H_{hx}-\bar{H}_{hy})t}\gamma_{0}\right)(x,x)\right\|_{L^2_{[-\pi/2b,\pi/2b]}L^2_{x}}\\
     =& \left\||\nabla_{x}|^{1/2+2\epsilon} \left(e^{i(\Delta_{x}-\Delta_{y})t}\gamma_{0}\right)(x,x)\right\|_{L^2_{t}L^2_{x}}.
\end{align*}

Finally, the estimate (\ref{pauli,lens,estimate}) reduces to the Laplacian case 
\begin{equation*}
    \left\||\nabla_{x}|^{1/2+2\epsilon} \left(e^{i(\Delta_{x}-\Delta_{y})t}\gamma_{0}\right)(x,x)\right\|_{L^2_{t}L^2_{x}}\lesssim_{\epsilon}\left\|\langle \nabla_x\rangle^{1/4+\epsilon}\langle \nabla_{y}\rangle^{1/4+\epsilon}\gamma_{0}(x,y)\right\|_{L^2_{x}L^2_{y}},
\end{equation*}
which is proved in \cite{GM17,CHP2017}.
\end{proof}
Since the Hermite operator $H_{h}$ dominates $-\Delta+1$ in the sense $\left\|\langle-\Delta\rangle^{s/2}f\right\|_{L^2}\lesssim \left\|H_{h}^{s/2}f\right\|_{L^2}$ for $s\ge 0$, as a corollary of Proposition \ref{collapsing,pauli,hermite}
\begin{equation}\label{collapsing,pauli,herimte,estimate}
    \left\|\langle\tan bt\rangle^{-1/2-\epsilon}|\nabla_{x}|^{1/2+2\epsilon} \rho_{\gamma}(t,x)\right\|_{L^2_{[-\pi/2b,\pi/2b]}L^2_{x}}\lesssim_{\epsilon}\left\|H_{hx}^{1/8+\epsilon}H_{hy}^{1/8+\epsilon}\gamma_{0}(x,y)\right\|_{L^2_{x}L^2_{y}}.
\end{equation}
using this estimate (\ref{collapsing,pauli,herimte,estimate}) and the scheme in \cite{GM17}, Theorem \ref{local,low,regularity,pauli} follows.

When it comes to Equation (\ref{equation,pertubation}), if we expect to establish a local well-posedness result when
\[
    \left\| H_{hx}^{s/2} H_{hy}^{s/2}Q_{0}(x,y)\right\|_{L^2_{x}L^2_{y}}<\infty,
\]
we need to deal with terms, for example $\left\|\left|\nabla_{x}\right|^{s}\left(\rho_{Q}*v\right) H_{hy}^{s/2}\bar{\Pi}_{\phi}(x,y)\right\|_{L^2_xL^2_{y}}$. However $H_{hy}^{s/2}\bar{\Pi}_{\phi}$ is not translation invariant. After integrating over $y$, we are faced with $\left\| |x|^{s}\left|\nabla_{x}\right|^{s}\left(\rho_{Q}*v\right)\right\|_{L^2}$. For the linear equation $i\,\partial_{t}Q=[H+\rho_{Q}*v,Q]$ and $Q(t=0)=Q_0$, $\left\| |x|^{s}\left|\nabla_{x}\right|^{s}\left(\rho_{Q}*v\right)\right\|_{L^2_{I_{T}}L^2}$ is controlled by $\left\| H_{hx}^{s/2} H_{hy}^{s/2}Q_{0}(x,y)\right\|_{tr}$.  But it may not be controlled by $\left\| H_{hx}^{s/2} H_{hy}^{s/2}Q_{0}(x,y)\right\|_{L^2_{x}L^2_{y}}$. Therefore we can not close the argument to obtain a local well-posedness result of Equation (\ref{equation,pertubation}). That is why we stick to the structure of Equation (\ref{equation,pertubation}) and use norms arising from $H$, i.e. Definition \ref{def,main,norm}. The operator $H$ is more compatible with the stationary solution $\bar{\Pi}_{\phi}$ than $H_{h}$. Hence we can deal with $\langle H_{x}\rangle^{1/2}\langle \bar{H}_{y}\rangle^{1/2}[\rho_{Q}*v,\bar{\Pi}_{\phi}]$.

Next we prove Theorem \ref{pauli,local,wellposed} the local wellposedness result of Equation (\ref{equation,pertubation}).
\begin{proof}
By Duhamel's formulation, we define the solution map $\Phi$ and the solution ball $sol_{T}$ for the contraction mapping principle,
\begin{align}
    \Phi(Q)(t,x,y):= e^{-i(H_{x}-\bar{H}_{y})t}Q_{0} -i\int_{0}^{t}e^{-i(H_{x}-\bar{H}_{y})(t-\tau)}[v*\rho_{Q},Q+\bar{\Pi}_{\phi}](\tau)\,d\tau,\\
    sol_{T}:=\left\{ Q(t,x,y)\left| \|Q(t,x,y)\|_{\mathbf{N}_{T}} \le C\left\|\langle H_{x}\rangle^{1/2}\langle \bar{H}_{y}\rangle^{1/2}Q_{0}(x,y)\right\|_{L_{x}^2L_{y}^2} \right.\right\},
\end{align}
where parameters $T$ and $C>1$ are to be determined later.

1. Show $\Phi$ maps $sol_{T}$ to itself.

Suppose $Q\in sol_{T}$. By Theorem \ref{pauli,collapsing,estimate} and Proposition \ref{Strichartz,homogeneous},
\[
    \|e^{-i(H_{x}-\bar{H}_y)t}Q_{0}\|_{\mathbf{N}_{T}}\lesssim_{T} \left\|\langle H_{x}\rangle^{1/2}\langle \bar{H}_{y}\rangle^{1/2}Q_{0}\right\|_{L_{x}^2L_{y}^2}.
\]
Choosing $T= \pi/4b$, then $C$ is the constant such that \[
    \|e^{-i(H_{x}-\bar{H}_y)t}Q_{0}\|_{\mathbf{N}_{T}}\le C\left\|\langle H_{x}\rangle^{1/2}\langle \bar{H}_{y}\rangle^{1/2}Q_{0}\right\|_{L_{x}^2L_{y}^2}/2.
\]

For the nonlinear part, claim the estimate
\begin{equation}\label{nonlinear,reduction}
    \left\|\int_{0}^{t}e^{-i(H_x-\bar{H}_{y})(t-\tau)}[v*\rho_{Q},Q+\bar{\Pi}_{\phi}](\tau)\,d\tau\right\|_{\mathbf{N}_{T}}\lesssim \left\|\langle H_{x}\rangle^{1/2} \langle \bar{H}_{y}\rangle^{1/2}[v*\rho_{Q},Q+\bar{\Pi}_{\phi}]\right\|_{L^1_{I_{T}}L^2_{x}L^2_{y}}.
\end{equation}
The proof of (\ref{nonlinear,reduction}) is twofold. On one hand, to control the Strichartz norm,
\[
    \left\|\int_{0}^{t}e^{-i(H_x-\bar{H}_{y})(t-\tau)}\underbrace{\langle H_{x}\rangle^{1/2} \langle \bar{H}_{y}\rangle^{1/2}[v*\rho_{Q},Q+\bar{\Pi}_{\phi}](\tau)}_{F_1(\tau,x,y)}\,d\tau\right\|_{L^{q}_{I_{T}}L^{r}_{x}L^2_{y}},
\]
suppose $G(t,x,y)$ is in the dual Strichartz space $L_{I_{T}}^{q'}L_{x}^{r'}L^2_{y}$, where
\[
    \frac{1}{q}+\frac{1}{q'}=1,\quad \frac{1}{r}+\frac{1}{r'}=1.
\]
Using the dual characterization of $L^p$ spaces
\begin{align*}
    &\int_{I_{T}}\int_{\mathbb{R}^2\times\mathbb{R}^2}dtdxdy\,\bar{G}(t,x,y)\int_{0}^{t}e^{-i(H_{x}-\bar{H}_y)(t-\tau)}F_1(\tau,x,y)d\tau\\
    =&\int_{I_{T}}\int_{\mathbb{R}^2\times\mathbb{R}^2}d\tau dxdy\,F_1(\tau,x,y)\int_{\tau}^{T}\overline{e^{-i(H_{x}-\bar{H}_y)(\tau-t)}G}(t,x,y)\,dt\\
    \le &\int_{I_{T}}d\tau \,\left\|F_1(\tau,x,y)\right\|_{L^{2}_x L^2_y} \left\|\int_{\tau}^{T}\overline{e^{-i(H_{x}-\bar{H}_y)(\tau-t)}G}(t,x,y)\,dt\right\|_{L^{\infty}_{I_{T}}L^2_{x}L^2_{y}}\\
    \lesssim& \int_{I_{T}}d\tau\,\left\|F_1(\tau,x,y)\right\|_{L^{2}_x L^2_y}\left\|G(t,x,y)\right\|_{L^{q'}_{I_{T}}L^{r'}_{x}L^2_y},\\
     & \quad (\text{by Proposition \ref{Strichartz,homogeneous} the dual Strichartz estimate (\ref{dual,strichartz})})
\end{align*}
we obtain, 
\[
    \left\|\int_{0}^{t}e^{-i\left(H_{x}-\bar{H}_{y}\right)(t-\tau)}F_1(\tau,x,y)\,d\tau\right\|_{L^q_{I_{T}}L^r_{x}L^2_{y}}\lesssim \left\|F_1(t,x,y)\right\|_{L^1_{I_{T}}L^2_{x}L^2_{y}}.
\]
The argument for the norm $L^{q}_{I_{T}}L^r_{y}L^2_{x}$ is the same. On the other hand, to control the collapsing term
\[
    \left\|\langle\nabla_{x}\rangle^{9/8}\left(\int_{0}^{t}e^{-i(H_x-\bar{H}_{y})(t-\tau)}\underbrace{[v*\rho_{Q},Q+\bar{\Pi}_{\phi}](\tau)}_{F_2(\tau,x,y)}\,d\tau\right)(t,x,x)\right\|_{L^2_{I_{T}}L^2_{x}},
\]
applying Theorem \ref{pauli,collapsing,estimate} and the Minkowski inequality,
\begin{align*}
    &\left\|\langle\nabla_{x}\rangle^{9/8}\left(\int_{0}^{t}e^{-i(H_x-\bar{H}_{y})(t-\tau)}F_2(\tau,x,y)\,d\tau\right)(t,x,x)\right\|_{L^2_{I_{T}}L^2_x}\\\le &\left\|\int_{0}^{T}\left\|\langle\nabla_{x}\rangle^{9/8} \left(e^{-i(H_x-\bar{H}_{y})(t-\tau)}F_2(\tau,x,y)\right)(t,x,x)\right\|_{L^2_x}d\tau\right\|_{L^2_{I_{T}}}\\
    \le & \int_{0}^{T}d\tau \left\|\langle\nabla_{x}\rangle^{9/8} \left(e^{-i(H_x-\bar{H}_{y})(t-\tau)}F_2(\tau,x,y)\right)(t,x,x)\right\|_{L^2_{I_{T}}L^2_x} \\
    \lesssim &\int_{0}^{T}d\tau \left\|\langle H_{x}\rangle^{1/2} \langle \bar{H}_{y}\rangle^{1/2}F_2(\tau,x,y)\right\|_{L^2_{x}L^2_{y}} \quad(\text{by Theorem \ref{pauli,collapsing,estimate}}).
\end{align*}

According to the estimate (\ref{nonlinear,reduction}), the problem is reduced to estimate quantities 
\[
    \left\|\langle H_{x}\rangle^{1/2} \langle \bar{H}_{y}\rangle^{1/2}[v*\rho_{Q},Q]\right\|_{L^1_{I_{T}}L^2_{x}L^2_{y}},\quad \left\|\langle H_{x}\rangle^{1/2} \langle \bar{H}_{y}\rangle^{1/2}[v*\rho_{Q},\bar{\Pi}_{\phi}]\right\|_{L^1_{I_{T}}L^2_{x}L^2_{y}}.
\]
Since the commutation relation does not play a role of our analysis, we give proofs for one of the two terms in the commutation relation. The other one is dealt similarly. 

Considering $\left\|\langle H_{x}\rangle^{1/2} \langle \bar{H}_{y}\rangle^{1/2}\left(\left(v*\rho_{Q}\right)(t,x)Q(t,x,y)\right)\right\|_{L^1_{I_{T}}L^2_{x}L^2_{y}}$, based on the observation (\ref{L2,equivalent}), we estimate it by
\begin{align*}
    &\lesssim \left\|D_{x}\langle \bar{H}_{y}\rangle^{1/2} \left(v*\rho_{Q}\right)(t,x)Q(t,x,y)\right\|_{L^1_{I_{T}}L^2_{x}L^2_{y}}+ \left\|\langle \bar{H}_{y}\rangle^{1/2} \left(v*\rho_{Q}\right)(t,x)Q(t,x,y)\right\|_{L^1_{I_{T}}L^2_{x}L^2_{y}}\\
    &\lesssim\left\| 2\partial_{z_{x}} \left(v*\rho_{Q}\right)(t,x)\, \langle \bar{H}_{y}\rangle^{1/2}Q(t,x,y)\right\|_{L^1_{I_{T}}L^2_x L^2_y} +\left\|\left(v*\rho_{Q}\right)(t,x)\langle H_{x}\rangle^{1/2}\langle \bar{H}_{y}\rangle^{1/2}Q(t,x,y)\right\|_{L^{1}_{I_{T}}L_x^{2}L_{y}^2}\\
    &\lesssim T^{1/2}\left\||\nabla_{x}|(v*\rho_{Q})(t,x)\right\|_{L^2_{I_{T}}L^{\frac{16}{7}}_x}\left\|\langle \bar{H}_{y}\rangle^{1/2}Q(t,x,y)\right\|_{L^{\infty}_{I_{T}}L^{16}_x L^2_y}\\
    &\quad +T^{1/4}\left\|\left(v*\rho_{Q}\right)(t,x)\right\|_{L^2_{I_{T}}L^4_{x}}\left\|\langle H_{x}\rangle^{1/2}\langle \bar{H}_{y}\rangle^{1/2}Q(t,x,y)\right\|_{L^{4}_{I_{T}}L_x^{4}L_{y}^2} \quad (\text{by H\"older inequality})\\
        &\lesssim T^{1/2}\left\||\nabla_{x}|^{9/8}\rho_Q(t,x)\right\|_{L^2_{I_{T}}L^2_{x}}\|\langle H_{x}\rangle^{1/2}\langle \bar{H}_{y}\rangle^{1/2}Q(t,x,y)\|_{L^{\infty}_{I_{T}}L^2_x L^2_y}\\
        &\quad + T^{1/4}\left\||\nabla_{x}|^{1/2}\rho_Q(t,x)\right\|_{L^2_{I_{T}}L^2_{x}}\left\|\langle H_{x}\rangle^{1/2}\langle \bar{H}_{y}\rangle^{1/2}Q(t,x,y)\right\|_{L^4_{I_{T}}L^4_x L^2_y}\\
        &\qquad(\text{by Sobolev inequality, Lemma \ref{covariant,n-endpoint,sobolev} and Young's convolution inequality})\\
        &\lesssim \max\{T^{1/2},T^{1/4}\}\left\|Q(t)\right\|_{\mathbf{N}_{T}}^2\\
        &\lesssim \max\{T^{1/2},T^{1/4}\}\left\|\langle H_{x}\rangle^{1/2}\langle \bar{H}_{y}\rangle^{1/2}Q_{0}(x,y)\right\|_{L_{x}^2L_{y}^2}^2.
\end{align*}
Then for $\left\|\langle H_{x}\rangle^{1/2} \langle \bar{H}_{y}\rangle^{1/2}[v*\rho_{Q},\bar{\Pi}_{\phi}]\right\|_{L^1_{I_{T}}L^2_{x}L^2_{y}}$, because of direct computation
\begin{align*}
    \bar{D}_{y}\bar{\Pi}_{\phi}(x,y)&=\left(2\partial_{z}\phi(x-y)+\frac{b}{2}(\bar{z}_{x}-\bar{z}_{y})\phi(x-y)\right)e^{-ib\Omega(x,y)/2},\\
    D_{x}\bar{\Pi}_{\phi}(x,y)&=\left(-2\partial_{\bar{z}}\phi(x-y)-\frac{b}{2}(z_{x}-z_{y})\phi(x-y)\right)e^{-ib\Omega(x,y)/2},\\
    D_{x}\bar{D}_{y}\bar{\Pi}_{\phi}(x,y)&=\left(-4\partial_{\bar{z}}\partial_{z}\phi(x-y)-b(z_{x}-z_{y})\partial_{z}\phi(x-y)\right.\\
        &\quad \left.-b(\bar{z}_{x}-\bar{z}_{y})\partial_{\bar{z}}\phi(x-y)-\frac{b^2}{4}|x-y|^2\phi(x-y)\right)e^{-ib\Omega(x,y)/2},
\end{align*}
integrating over $x$ or $y$, we obtain
\begin{align*}
    &\left\|\bar{D}_{y}\bar{\Pi}_{\phi}(x,y)\right\|_{L^2_{x(y)}}=\left\|\bar{D}\phi\right\|_{L^2}\lesssim\left\|\langle\bar{H}\rangle^{1/2}\phi\right\|_{L^2}\\
    &\left\|D_{x}\bar{\Pi}_{\phi}(x,y)\right\|_{L^2_{x(y)}}=\left\|D\phi\right\|_{L^2}\lesssim\left\|\langle H\rangle^{1/2}\phi\right\|_{L^2},\\
    &\left\|D_{x}\bar{D}_{y}\bar{\Pi}_{\phi}(x,y)\right\|_{L^2_{x(y)}}=\left\|D\bar{D}\phi\right\|_{L^2}\lesssim\left\|\langle H\rangle^{1/2} \langle\bar{H}\rangle^{1/2}\phi\right\|_{L^2}.
\end{align*}
Combining the above estimates,
\begin{align*}
    \left\|\langle H_{x}\rangle^{1/2} \langle \bar{H}_{y}\rangle^{1/2}[v*\rho_{Q},\bar{\Pi}_{\phi}]\right\|_{L^1_{I_{T}}L^2_{x}L^2_{y}}&\lesssim \left\|\langle\nabla_{x}\rangle \rho_{Q}(t,x)\right\|_{L^1_{I_{T}}L^2_{x}} \left\|\langle H\rangle^{1/2}\langle \bar{H}\rangle^{1/2}\phi\right\|_{L^2}\\
     &\lesssim T^{1/2}\left\|\langle H_{x}\rangle^{1/2}\langle \bar{H}_{y}\rangle^{1/2}Q_{0}(x,y)\right\|_{L_{x}^2L_{y}^2}\left\|\langle H\rangle^{1/2}\langle \bar{H}\rangle^{1/2}\phi\right\|_{L^2}
\end{align*}

If necessary, shrink the interval $I_{T}$ such that
\[
    \left\|\int_{0}^{t}e^{-i(H_x-\bar{H}_{y})(t-\tau)}[v*\rho_{Q},Q+\bar{\Pi}_{\phi}](\tau)\,d\tau\right\|_{\mathbf{N}_{T}}\le \frac{C}{2}\left\|\langle H_{x}\rangle^{1/2}\langle \bar{H}_{y}\rangle^{1/2}Q_{0}(x,y)\right\|_{L_{x}^2L_{y}^2}.
\]
Thus $\Phi$ maps $sol_{T}$ to itself.

2. Show $\Phi$ is a contraction map.

For any $Q_{1},\, Q_{2}\in sol_{T}$, similarly as step 1, 
\begin{align*}
    \left\|\Phi(Q_{1})-\Phi(Q_{2})\right\|_{\mathbf{N}_{T}} &\le \left\|\int_{0}^{t}d\tau\, e^{-i(H_{x}-\bar{H}_{y})(t-\tau)} [v*\rho_{Q_1}-v*\rho_{Q_2},\bar{\Pi}_{\phi}]\right\|_{\mathbf{N}_{T}}\\
        &\quad+\left\|\int_{0}^{t}d\tau\, e^{-i(H_{x}-\bar{H}_{y})(t-\tau)} [v*\rho_{Q_1}-v*\rho_{Q_2},Q_1]\right\|_{\mathbf{N}_{T}}\\
     &\quad +\left\|\int_{0}^{t}d\tau\, e^{-i(H_{x}-\bar{H}_{y})(t-\tau)} [v*\rho_{Q_2},Q_1-Q_2]\right\|_{\mathbf{N}_{T}}\\
     &\lesssim T^{1/2}\|\langle \nabla_{x}\rangle (\rho_{Q_1}-\rho_{Q_2})\|_{L^2_{I_{T}}L^2_{x}}\left\|\langle H\rangle^{1/2}\langle \bar{H}\rangle^{1/2}\phi\right\|_{L^2}\\
     &\quad +\max\{T^{1/2},T^{1/4}\}\|Q_{1}-Q_{2}\|_{\mathbf{N}_{T}}\left(\|Q_1\|_{\mathbf{N}_{T}}+\|Q_2\|_{\mathbf{N}_{T}}\right)\\
     &\lesssim \max\{T^{1/2},T^{1/4}\}\|Q_{1}-Q_{2}\|_{\mathbf{N}_{T}}\left\|\langle H_{x}\rangle^{1/2}\langle \bar{H}_{y}\rangle^{1/2}Q_{0}(x,y)\right\|_{L_{x}^2L_{y}^2}.
\end{align*}
If needed, choose a smaller $T$ such that $\left\|\Phi(Q_{1})-\Phi(Q_{2})\right\|_{\mathbf{N}_{T}}\le \|Q_1-Q_2\|_{\mathbf{N}_{T}}/2$.

Then by the contraction mapping principle, $\Phi$ has a fixed point in $sol_{T}$, i.e. Equation (\ref{equation,pertubation}) is locally well-posed.
\end{proof}
\begin{remark}
There are two families of stationary solutions $\Pi_{\phi}$ and $\bar{\Pi}_{\phi}$ (see Section \ref{stationary,solutions}). The reason for only $\bar{\Pi}_{\phi}$ is used in our pertubation problem is twofold. On one hand, $\bar{\Pi}_{\phi}$ recovers the Fermi-Dirac distribution. On the other hand, suppose we use the stationary solution $\Pi_{\phi}$ instead of $\bar{\Pi}_{\phi}$. By the product rule of the covariant derivative $D$, $D(fg)=(Df)g-2f\partial_{\bar{z}}g$,
\begin{align}
    & D_{x}\bar{D}_{y}\left(\rho_{u}*v(x)\Pi_{\phi}(x,y)\right)=D_{x}\left(\rho_{u}*v\right)(x)\bar{D}_{y}\Pi_{\phi}(x,y)+\left(\rho_{u}*v\right)(x)\left(-2\partial_{\bar{z}_{x}}\right)\bar{D}_{y}\Pi_{\phi}(x,y)\nonumber\\
    \hbox{or }& D_{x}\bar{D}_{y}\left(\rho_{u}*v(x)\Pi_{\phi}(x,y)\right)=(-2\partial_{\bar{z}_{x}})\left(\rho_{u}*v\right)(x)\bar{D}_{y}\Pi_{\phi}(x,y)+\left(\rho_{u}*v\right)(x)D_{x}\bar{D}_{y}\Pi_{\phi}(x,y).\label{stationary,comparision}
\end{align}
Since we do not have an estimate for $D_{x}\rho_{u}(t,x)$, we use the form (\ref{stationary,comparision}) to continue our argument. A direct computation shows 
\[
    \bar{D}_{y}\Pi_{\phi}(x,y)=\left(2\partial_{z}\phi(x-y)-\frac{b}{2}(\bar{z}_{x}+\bar{z}_{y})\phi(x-y)\right)e^{ib\Omega(x,y)/2}.
\]
$\left|\bar{D}_{y}\Pi_{\phi}(x,y)\right|$ is not translation invariant. Therefore in order to estimate
\[
    \left\|-2\partial_{\bar{z}_{x}}(\rho_{u}*v)(t,x)\bar{D}_{y}\Pi_{\phi}(x,y)\right\|_{L^2_{x}L^2_{y}},
\]
we need to control $\left\|x|\nabla_{x}|\rho_{u}(t,x)\right\|_{L^2_{x}}$, which is not possible by using $\mathbf{N}_{T}$. 
\end{remark}

\section{Conclusion}\label{conclusion}
In this paper, we obtained a local well-posed result of Equation (\ref{equation,pertubation}) and a new collapsing estimate Theorem \ref{pauli,collapsing,estimate}. However the estimate is not sharp since we do not have an optimal control of associated Laguerre polynomials (see Remark \ref{associated,Laguerre,not,optimal}).

The ultimate goal of Theorem \ref{pauli,local,wellposed} is to acquire a low regularity result, for example a local well-posedness result for the initial data
\[
    \left\|\langle H_{x}\rangle^{s/2}\langle \bar{H}_{y}\rangle^{s/2}Q_{0}(x,y)\right\|_{L^2_{x}L^2_{y}}<\infty,\quad s<1,
\]
According to Remark \ref{associated,Laguerre,not,optimal} and the proof of Theorem \ref{pauli,collapsing,estimate}, we have a little gain of derivatives for the collapsing term when $s>1/3$. We conjecture that the best case might be $s=1/3+\epsilon$.  However it requires a fractional Leibniz rule for $\langle H\rangle^{s/2}(fg)$, which currently is beyond our ability. 

Another direction is to establish a global well-posedness result when
\[
    \left\|\langle H_{x}\rangle^{1/2}\langle \bar{H}_{y}\rangle^{1/2}Q_{0}(x,y)\right\|_{Tr}<\infty.
\]
A formal computation shows that the total energy of Equation (\ref{equation,pertubation}) is conserved 
\begin{equation}
    \mathcal{E}(Q)=Tr\left(H^{1/2}QH^{1/2}\right)+\frac{1}{2}\int_{\mathbb{R}^2}\left(v*\rho_{Q}\right)(x)\rho_{Q}(x)\,dx,
\end{equation}
which can be used for the global well-posedness result. However we lack of tools to estimate $\left\|\rho_{Q}\bar{\Pi}_{\phi}\right\|_{Tr}$ by the initial data.

\section{Appendix}
\subsection{Heisenberg Group}\label{Heisenberg,group}
\cite[Chapter 1]{Folland89}Let us review the Heisenberg group $H_{1}$ with the group law
\begin{equation*}
    (p_1,q_1,t_1)\cdot (p_2,q_2,t_2)=\left(p_1+p_2,q_1+q_2,t_1+t_2+b\frac{\Omega\left((p_1,q_1),(p_2,q_2)\right)}{2}\right),
\end{equation*}
where $p_{i},q_{i}\in \mathbb{R}$, $t_{i}\in\mathbb{R}$, and impose a complex structure on $\mathbb{R}^{2}$, $z=p+q\,i$. 

Identify the tangent space $TH_{1}$ with $\mathbb{R}^{3}\times\mathbb{R}^3$ and its basis by $\left\{\partial_{p},\partial_{q},\partial_{t}\right\}$. Then the differential of the left multiplication $L_{g}$, where $g=(p_{g},q_{g},t_{g})$, is
\begin{align*}
    DL_{g}\left(\partial_{p},\partial_{q},\partial_{t}\right)&=\left(\partial_{p},\partial_{q},\partial_{t}\right)\begin{pmatrix}
        1 & 0 & 0\\
        0 & 1 & 0\\
        -bq_{g}/2 & bp_{g}/2 & 1
    \end{pmatrix}.
\end{align*}
The Lie algebra $\mathfrak{h}_{1}$ consisting of left invariant vector fields is
\begin{align*}
    \mathfrak{h}_{1}&=\mathbb{R}\hbox{-span}\left\{ \partial_{p}-b\frac{q}{2}\partial_{t},\ \partial_{q}+b\frac{p}{2}\partial_{t},\ \partial_{t}\right\},
\end{align*}
and the corresponding complexified space is
\begin{align*}
    \mathfrak{h}_{1}^{\mathbb{C}}&=\mathbb{C}\hbox{-span}\left\{ 2\partial_{\bar{z}}+i\frac{bz}{2}\partial_{t},\ 2\partial_{z}-i\frac{b\bar{z}}{2}\partial_{t},\ \partial_{t}\right\}.
\end{align*}

We will think of $D$ and $D^*$ as vector fields of $\mathfrak{h}_{1}^{\mathbb{C}}$ in the following way. Denote 
\[
    D_{H_{1}}=-2\partial_{\bar{z}}-i\frac{bz}{2}\partial_{t},\quad D^*_{H_{1}}=2\partial_{z}-i\frac{b\bar{z}}{2}\partial_{t}.
\] 
Suppose $\tilde{f}\in \mathcal{S}(H_{1})$ and apply the inverse Fourier transform on $t$ variable,
\[
    \check{D}_{H_{1}}\check{\tilde{f}}=\frac{1}{\sqrt{2\pi}}\int_{\mathbb{R}}\left(-2\partial_{\bar{z}}-\frac{bz\tau}{2}\right)\tilde{f}(q,p,t)e^{it\tau}\,dt.
\]
On the piece $\tau=1$, $D_{H_{1}}$ and $D^*_{H_{1}}$ correspond to $D$ and $D^{*}$ respectively.

To make this correspondence rigorous, consider a quotient group $H_{1}^{red}$ of $H_{1}$ 
\begin{equation*}
    H_{1}^{red}:=H_{1}/\left\{\left.(0,0,t)\right| t\in 2\pi\mathbb{Z}\right\},\ \left\{\left.(0,0,t)\right| t\in 2\pi\mathbb{Z}\right\}\subset C(H_{1}).
\end{equation*}
For a $f$ on $\mathbb{R}^{2}$, it is lifted to $H_{1}$ by defining
\begin{equation}\label{R,H,correspondence}
    \tilde{f}(p,q,t):=\sqrt{2\pi}\exp(-ti)f(p,q).
\end{equation}
Through the definition (\ref{R,H,correspondence}), the correspondence between $D(D^*)$ and $D_{H_{1}}\left(D_{H_{1}}^*\right)$ is 
\begin{equation}\label{D,Lie algebra}
    D\tilde{f}(p,q,t)=D^*_{H_1}\tilde{f}(p,q,t), \quad D^*\tilde{f}(p,q,t)=D^*_{H_1}\tilde{f}(p,q,t).
\end{equation}
We can also relate the twisted convolution defined in (\ref{twisted,convolution,def}) to the group convolution on $H_{1}$,
\begin{align*}
    \left(\tilde{f}*\tilde{g}\right)(p,q,t) &=\int_{H_{1}^{red}}\tilde{f}\left((p,q,t)\cdot(\tilde{p},\tilde{q},\tilde{t})^{-1}\right)\tilde{g}(\tilde{p},\tilde{q},\tilde{t})\,d\tilde{p}d\tilde{q}d\tilde{t}\\
    &=\int_{H_{1}^{red}}\tilde{f}\left(p-\tilde{p},q-\tilde{q},t-\tilde{t}-b\frac{\Omega\left((p,q),(\tilde{p},\tilde{q})\right)}{2}\right)\tilde{g}(\tilde{p},\tilde{q},\tilde{t})\,d\tilde{p}d\tilde{q}d\tilde{t}\\
        &=2\pi\exp(-t i) \int_{\mathbb{R}^{2}}f(p-\tilde{p},q-\tilde{q})g(\tilde{p},\tilde{q})\exp\left(ib\frac{\Omega\left((p,q),(\tilde{p},\tilde{q})\right)}{2}\right)\,d\tilde{p}d\tilde{q}\\
        &=2\pi\exp(-ti)\left(f\natural g\right)(p,q),
\end{align*}
i.e. $\tilde{f}*\tilde{g}=\sqrt{2\pi}\widetilde{f\natural g}$.

\begin{lemma}\label{commutation,Lx,convolution}
Let $G$ be a Lie group endowed with a left invariant Haar measure $d\mu$, then
\begin{equation}\label{Lie group,commutation,convolution}
    L_{X}\left(f*g\right)=f*L_{X}g,\quad X\in \mathfrak{g}
\end{equation}
where $L_{X}$ denotes the Lie derivative by X and $*$ denotes the convolution on $G$
\[
    \left(f*g\right)(x):=\int_{G}f(xy^{-1})g(y)dy,\quad x\in G.
\]
Furthermore, (\ref{Lie group,commutation,convolution}) holds for the complexified Lie algebra $\mathfrak{g}_{\mathbb{C}}$. 
\end{lemma}
\begin{proof}
Suppose $X\in \mathfrak{g}$, let $\exp(tX)$ denote the one parameter subgroup generated by $X$ and $\exp(tX).x$ denote the action of $\exp(tX)$ on $G$, i.e. $x\in G$ travels along the flow generated by $X$. Then
\begin{align*}
    \int_{G}f\left((\exp(tX).x)y^{-1}\right)g(y)\,dy &=\int_{G}f\left(x\exp(tX)y^{-1}\right)g(y)\,dy \\
        &=\int_{G}f\left(x\left(y\exp(-tX)\right)^{-1}\right)g(y)\,dy \\
        &=\int_{G}f\left(xy^{-1}\right) g\left(\exp(tX).y\right)\,L^*_{\exp(tX)}dy \\
        &=\int_{G}f\left(xy^{-1}\right) g\left(\exp(tX).y\right)\,dy 
\end{align*}
which implies the identity (\ref{Lie group,commutation,convolution}).
\end{proof}

\subsection{Stationary Solutions}\label{stationary,solutions}
We use relations (\ref{D,Lie algebra}) to find two families of stationary solutions to Equation (\ref{hatree,density,form,main}).
\begin{proposition}\label{stationary,solution}
Suppose $v\in L^1\left(\mathbb{R}^2\right)$, there are two families of stationary solutions to Equation (\ref{hatree,density,form,main}),
\begin{enumerate}[label=(\roman*)]
    \item $\displaystyle \Pi_{\phi}(x,y)=\phi(x-y)\exp\left(ib\frac{\Omega(x,y)}{2}\right)$, for arbitrary $\phi$ on $\mathbb{R}^{2}$,
    \item $\displaystyle\bar{\Pi}_{\phi}(x,y)=\phi(x-y)\exp\left(-ib\frac{\Omega(x,y)}{2}\right)$, where $\phi$ is of radial symmetry, i.e. $\phi(x)=\phi(|x|)$.
\end{enumerate}
\end{proposition}
\begin{proof}
By the correspondence (\ref{D,Lie algebra}), we regard $D$ and $D^*$ as vector fields of $H_1$. 
Since the Lebesgue measure on $H_{1}^{red}$ is bi-invariant and the group convolution on $H_1$ is related to the twisted convolution by $\tilde{f}*\tilde{g}=\sqrt{2\pi}\widetilde{f\natural g}$, using Lemma \ref{commutation,Lx,convolution}, we conclude that the Hamiltonian $H=D^*D$ commutes with the twisted convolution $\natural$. As a result, 
\begin{align*}
        & [D^*D,\Pi_{\phi}]=D^*[D,\Pi_{\phi}]+[D^{*},\Pi_{\phi}]D=0\\
    \implies & H_{x}\int_{\mathbb{R}^{2}}\Pi_{\phi}(x,y)f(y)dy-\int_{\mathbb{R}^{2}}\Pi_{\phi}(x,y)H_{y}f(y)dy\\
        & =\int_{\mathbb{R}^{2}}\left(H_{x}\Pi_{\phi}(x,y)-\bar{H}_{y}\Pi_{\phi}(x,y)\right)f(y)dy=0,\ \forall f\in\mathcal{S}(\mathbb{R}^{2})
\end{align*}
Besides $\Pi_{\phi}(x,x)=\phi(0)$ and $v*\phi(0)=\phi(0)\int v(x)\,dx$ are constant, $\Pi_{\phi}$ is a stationary solution to (\ref{hatree,density,form,main}).

Meanwhile, if we calculate $\left(H_{x}-\bar{H}_{y}\right)\bar{\Pi}_{\phi}$ directly, 
\begin{align*}
    \left(H_{x}-\bar{H}_{y}\right)\bar{\Pi}_{\phi}&=\left(H_{x}-\bar{H}_{x}-\bar{H}_{y}+H_{y}\right)\bar{\Pi}_{\phi}+\left(\bar{H}_{x}-H_{y}\right)\bar{\Pi}_{\phi}\\
        &=2ib\left(xJ\nabla_{x}+yJ\nabla_{y}\right)\bar{\Pi}_{\phi}\\
        &=2ib(x-y)^{T}J\left(\nabla_{x}\bar{\phi}(x-y)\right)\exp\left(-\frac{ib\Omega(x,y)}{2}\right),
\end{align*}
which vanishes if $\phi$ is a function of radial symmetry.
\end{proof}

\subsection{Transform}\label{transform}
We list some important results about the Fourier-Wigner transform $V$ and the Wigner transform $W$ from \cite[Chapter 1]{Folland89}. In the paper, we choose the reduced Planck constant $\hbar$ in \cite[Chapter 1]{Folland89} to be $b$ and use the following results when the dimension $d=1$.
\begin{proposition}{\cite[Proposition 1.42]{Folland89}}
\[
    \left\langle V(f_1,g_1),V(f_2,g_2) \right\rangle=\left(\frac{2\pi}{b}\right)^{d}\left\langle f_1,f_2\right\rangle\left\langle g_2,g_1\right\rangle,\quad f_{j},g_{j}\in L^2(\mathbb{R}^d),\ j=1,2.
\]
\end{proposition}

\begin{proposition}{\cite[Proposition 1.47]{Folland89}}
Suppose $f_{j},g_{j}\in L^2(\mathbb{R}^d)$,
        \begin{equation*}
            \overline{V(f_1,g_1)}\natural\overline{V(f_2,g_2)}=\left(\frac{2\pi}{b}\right)^{d}\langle g_2,f_1\rangle \overline{V(f_2,g_1)}.
        \end{equation*}
\end{proposition}

\begin{proposition}{\cite[Proposition 1.94]{Folland89}}\label{V,W,properties}
\begin{align*}
    W\left(\beta(a,e)f,\beta(c,d)g\right)(\xi,x)
    &=\exp\left(-\frac{ib}{2}\Omega\left((a,e),(c,d)\right)+i\left\langle (a,e)-(c,d),(\xi,x)\right\rangle\right)\\
    &\quad \cdot W(f,g)\left(\xi-\frac{b(e+d)}{2},x+\frac{b(a+c)}{2}\right).
    \end{align*}
where $a,e,c,d,x,\xi\in\mathbb{R}^d$.
\end{proposition}

Hermite functions and associated Laguerre polynomials are related by the following two theorems.
\begin{theorem}{\cite[Theorem 1.104]{Folland89}}\label{V,Hermit,Laguerre}
Suppose $p,q\in\mathbb{R}$, and $w=p+iq$. Then
\begin{equation*}
    V(h_{j},h_{k})(p,q)=\left\{
    \begin{array}{lc}
        \displaystyle\sqrt{\frac{k!}{j!}}\left(\sqrt{\frac{b}{2}}w\right)^{j-k}e^{-b|w|^2/4}L^{j-k}_{k}\left(\frac{b|w|^2}{2}\right),& j\ge k \\
        \displaystyle (-1)^{j+k}\sqrt{\frac{j!}{k!}}\left(\sqrt{\frac{b}{2}}\bar{w}\right)^{k-j}e^{-b|w|^2/4}L^{k-j}_{j}\left(\frac{b|w|^2}{2}\right),& j\le k  
    \end{array}
    \right.
\end{equation*}
\end{theorem}

\begin{theorem}{\cite[Theorem 1.105]{Folland89}}\label{W,Hermite,Laguerre}
Suppose $x,\xi\in\mathbb{R}$ and $z=x+i\xi$. Then
\begin{equation*}
    W(h_j,h_k)(\xi,x)=\left\{\begin{array}{cc}
        \displaystyle (-1)^k\frac{2}{b}\sqrt{\frac{k!}{j!}}\left(\sqrt{\frac{2}{b}}\bar{z}\right)^{j-k}L^{j-k}_{k}\left(\frac{2|z|^2}{b}\right)e^{-|z|^2/b}, & j\ge k \\
        \displaystyle (-1)^j\frac{2}{b}\sqrt{\frac{j!}{k!}}\left(\sqrt{\frac{2}{b}}z\right)^{k-j}L_{j}^{k-j}\left(\frac{2|z|^2}{b}\right)e^{-|z|^2/b}, & j\le k
    \end{array}
    \right.
\end{equation*}
\end{theorem}

Let $\mu$ be the Metaplectic representation from $Mp(2d,\mathbb{R})$ to $U\left(L^2(\mathbb{R}^d)\right)$, with infinitesimal representation 
\[
    d\mu: \mathcal{A}=\begin{pmatrix}A & B\\ C& -A^{T}\end{pmatrix}\in \mathfrak{sp}(2d,\mathbb{R}) \mapsto -\frac{1}{2i}\begin{pmatrix}\hat{Q} & \hat{P}\end{pmatrix}\begin{pmatrix}
    A & B\\
    C & -A^{T}
    \end{pmatrix}\begin{pmatrix}0 & id\\ -id & 0 \end{pmatrix}
    \begin{pmatrix}\hat{Q} \\ \hat{P}\end{pmatrix},
\]
where $\hat{Q}=x$, $\hat{P}=-i\nabla_{x}$, $x\in\mathbb{R}^d$ and $id$ is the identity matrix on $\mathbb{R}^d$.
\begin{theorem}{\cite[Theorem 4.51]{Folland89}}\label{metaplectic,expression}
Suppose $\begin{pmatrix}
        A(t) & B(t)\\
        C(t) & D(t)\\
    \end{pmatrix}=\exp\left(\begin{pmatrix}
    A & B\\
    C & -A^{T}
    \end{pmatrix}t\right)$, where $\begin{pmatrix}
    A & B\\
    C & -A^{T}
    \end{pmatrix}\in\mathfrak{sp}(2d,\mathbb{R})$. For any time $T>0$ such that when $t\in[0,T]$, $\det(D(t))>0$, then
\begin{align}
    &\quad \mu\begin{pmatrix}
        A(t) & B(t)\\
        C(t) & D(t)\\
    \end{pmatrix}f(x)\\
    &=\frac{1}{\det\left(D(t)\right)^{1/2}(2\pi)^{n/2}}\int_{\mathbb{R}^n}\exp\left(-iS(x,\xi)\right)\hat{f}(-\xi)\,d\xi,\ x\in\mathbb{R}^n,\ t\in[0,T]\nonumber
\end{align}
where
\[
    S(x,\xi)=\frac{-\xi D(t)^{-1}C(t)\xi}{2}+\xi D(t)^{-1}x+\frac{xB(t)D(t)^{-1}x}{2},\quad x,\xi\in\mathbb{R}^d.
\]
\end{theorem}

\subsection{Global Well-posedness}\label{global,trace,class}
We establish a global well-posedness result for Equation (\ref{hatree,density,form}) when 
\[
    \left\|h^{1/2}\Gamma_{0}h^{1/2}\right\|_{tr}<\infty,\quad \Gamma_{0}^*=\Gamma_{0},\quad \Gamma_{0}\ge 0 \text{ and } w(x)=\frac{1}{|x|}.
\]
The associated total energy is 
\begin{equation}\label{total,energy}
    \mathcal{E}_{HF}\left(\Gamma(t)\right)=Tr\left(h^{1/2}\Gamma(t)h^{1/2}\right)+\frac{1}{2}\int_{\mathbb{R}^3}\left(\rho_{\Gamma}*w\right)(t,x)\rho_{\Gamma}(t,x)\,dx.
\end{equation}
The outline of the proof is that we first establish two local well-posedness results for Equation (\ref{hatree,density,form}): one is at the energy level and another one is for smooth data. Then we verify the conservation law of the total energy for smooth data and use a limiting argument to pass the law to the energy level. Finally, the global well-posedness follows from the conservation of energy.
All estimates involved are based on time-independent arguments. 

Note that $h=L^*L$, where $L=\left(-i\partial_{x^1}+\frac{b}{2}x^2,-i\partial_{x^2}-\frac{b}{2}x^1,-i\partial_{x^3}\right)$ and $x=(x^1,x^2,x^3)$, and the covariant derivative $L$ is metric. The pointwise Kato's inequality holds
\begin{equation}\label{covariant,derivative,relation}
    \left|\nabla|f|\right|\lesssim \left|Lf\right|.
\end{equation}

Let us define the following operator norms for the discussion
\begin{equation}
    \left\|\Gamma\right\|_{\mathcal{L}^{s,p}}:=\left\|h^{s/2}\Gamma h^{s/2}\right\|_{\mathcal{L}^p}=\left(Tr\left|h^{s/2}\Gamma h^{s/2}\right|^{p}\right)^{1/p}
\end{equation}
where $s\ge0$, $1\le p\le\infty$ and $\mathcal{L}^p$ is the $p$-th Schatten norm.

1. The local well-posedness at the energy level.

To deal with the nonlinear term in Equation (\ref{hatree,density,form}), we first show a bilinear estimate for functions, then generalize it to operators.
\begin{proposition}\label{estimate,local,energy,level}
\begin{equation}
    \left\|h^{1/2}\left(\left(|\phi_1|^2*w\right) \phi_2\right)\right\|_{L^2} \lesssim \left\|h^{1/2}\phi_{1}\right\|_{L^2}^2\left\|h^{1/2}\phi_{2}\right\|_{L^2}.
\end{equation}
\end{proposition}
\begin{proof}
Applying the H\"older inequality, 
\begin{align*}
    \left\|h^{1/2}\left(\left(|\phi_1|^2*w\right) \phi_2\right)\right\|_{L^2}&\lesssim \left\||\phi_1|^2*w\right\|_{L^{\infty}}\left\|h^{1/2}\phi_2\right\|_{L^2}+\left\||\nabla_{x}|\left(|\phi_1|^2*w\right)\right\|_{L^3}\left\|\phi_2\right\|_{L^6},
\end{align*}
while 
\begin{align*}
    \left(|\phi_1|^2*w\right)(x)&=\int_{\mathbb{R}^3}|\phi_1|^2(x-y)w(y)\,dy \\
        &\lesssim \int_{\mathbb{R}^3}\left||\nabla_{y}|^{1/2}|\phi_1|(x-y)\right|^2\,dy \quad(\text{by the Hardy's inequality})\\
        &\le \int_{\mathbb{R}^3}\left(\left| |\nabla||\phi_1|\right|^2(x) +|\phi_1|^2(x)\right)\,dx \\
        &\lesssim \left\|h^{1/2}\phi_1\right\|_{L^2}^2, \quad(\text{by the inequality (\ref{covariant,derivative,relation})})
\end{align*}
and by the inequality (\ref{covariant,derivative,relation}), the Sobolev inequality and the Hardy-Littlewood-Sobolev inequality,
\begin{align*}
    &\left\|\phi_2\right\|_{L^6}\lesssim \left\||\phi_2|\right\|_{H^1}\lesssim \left\|h^{1/2}\phi_2\right\|_{L^2},\\
    &\left\||\nabla_{x}|\left(|\phi_1|^2*w\right)\right\|_{L^3}=\left\||\phi_1|^2*\left(|\nabla|w\right)\right\|_{L^3}\lesssim \left\||\phi_1|^2\right\|_{L^{3/2}}=\left\|\phi_1\right\|_{L^3}^2\lesssim \left\|h^{1/2}\phi_1\right\|^2_{L^2},
\end{align*}
we obtain the desired estimate,
\begin{equation*}
    \left\|h^{1/2}\left(\left(|\phi_1|^2*w\right) \phi_2\right)\right\|_{L^2} \lesssim \left\|h^{1/2}\phi_{1}\right\|_{L^2}^2\left\|h^{1/2}\phi_{2}\right\|_{L^2}.
\end{equation*}
\end{proof}

\begin{proposition}\label{estimate,operator,local,energy,level}
Suppose $\Gamma_1$ and $\Gamma_{2}$ are self-adjoint, 
\begin{equation}
    \left\|\left[\rho_{\Gamma_{1}}*w,\Gamma_{2}\right]\right\|_{\mathcal{L}^{1,1}}\lesssim \left\|\Gamma_{1}\right\|_{\mathcal{L}^{1,1}}\left\|\Gamma_{2}\right\|_{\mathcal{L}^{1,1}}
\end{equation}
\end{proposition}
\begin{proof}
Since $\Gamma_{j}$ is self-adjoint and $\left\|\Gamma_{j}\right\|_{\mathcal{L}^{1,1}}<\infty$ for $j=1,2$, there are orthonormal bases $\{f_{k,j}\}_{k=1}^{\infty}$ $j=1,2$, such that
\[
    \left(h^{1/2}\Gamma_{j}h^{1/2}\right)(x,y)=\sum_{k=1}^{\infty}\lambda_{k,j}f_{k,j}(x)\bar{f}_{k,j}(y).
\] 
Then
\[
   \Gamma_{j}(x,y)=\sum_{k=1}^{\infty}\lambda_{k,j}\left(h^{-1/2}f_{k,j}\right)(x)\left(\overline{h^{-1/2}f}_{k,j}\right)(y), 
\]
and by the Minkowski's inequality,
\begin{align*}
    \left\|\left(\rho_{\Gamma_{1}}*w\right)\Gamma_{2}\right\|_{\mathcal{L}^{1,1}} &=\left\|h^{1/2}_{x}\left(\left(\rho_{\Gamma_{1}}*w\right)(x) \sum_{k=1}^{\infty}\lambda_{k,2}\left(h^{-1/2}f_{k,2}\right)(x)\left(\bar{f}_{k,2}\right)(y)\right)\right\|_{tr}\\
    &\le\sum_{k=1}^{\infty}|\lambda_{k,2}|\left\|h^{1/2}_{x}\left(\left(\rho_{\Gamma_{1}}*w\right)(x)\left(h^{-1/2}f_{k,2}\right)(x)\left(\bar{f}_{k,2}\right)(y)\right)\right\|_{tr}\\
    &\le \sum_{k=1}^{\infty}|\lambda_{k,2}|\left\|h^{1/2}_{x}\left(\left(\rho_{\Gamma_{1}}*w\right)(x)\left(h^{-1/2}f_{k,2}\right)(x)\right)\right\|_{L^2_{x}}\\
    &\le \sum_{k=1}^{\infty}|\lambda_{k,2}|\sum_{l=1}^{\infty}|\lambda_{l,1}|\left\|h^{1/2}\left(\left(\left|h^{-1/2}f_{l,1}\right|^2*w\right)h^{-1/2}f_{k,2}\right)\right\|_{L^2}\\
    &\lesssim \sum_{k=1}^{\infty}|\lambda_{k,2}|\sum_{l=1}^{\infty}|\lambda_{l,1}| \left\|f_{l,1}\right\|_{L^2}^2\left\|f_{k,2}\right\|_{L^2}\quad (\text{by Proposition \ref{estimate,local,energy,level}})\\
    &\le \sum_{k=1}^{\infty}|\lambda_{k,2}|\sum_{l=1}^{\infty}|\lambda_{l,1}|.
\end{align*}
The other term $\left\|\Gamma_{2}\left(\rho_{\Gamma_{1}}*w\right)\right\|_{\mathcal{L}^{1,1}}$ can be estimated in the same way.
\end{proof}

Based on Proposition \ref{estimate,operator,local,energy,level}, we obtain the following local well-posedness result as an application of the contraction mapping principle.
\begin{theorem}\label{local,energy,level}
For any initial data $\left\|\Gamma_{0}\right\|_{\mathcal{L}^{1,1}}<\infty$ and $\Gamma_{0}^*=\Gamma_0$, Equation (\ref{hatree,density,form}) has a mild solution in the Banach space $\mathbf{N}_{1T}$, where the norm $\mathbf{N}_{1T}$ is defined as 
\begin{equation}
    \left\| \Gamma(t)\right\|_{\mathbf{N}_{1T}} :=\left\| \Gamma(t)\right\|_{L^{\infty}\left([0,T];\mathcal{L}^{1,1}\right)},
\end{equation}
while the existence time $T$ depends on $\left\|h^{1/2}\Gamma_{0}h^{1/2}\right\|_{tr}$. To be more precise, the solution $\Gamma(t)\in C^{0}\left([0,T];\mathcal{L}^{1,1}\right)$.
\end{theorem}

2. The local well-posedness for smooth data.

Similarly as Step 1,  we first show a bilinear estimate for functions, then generalize it to operators.
\begin{proposition}\label{estimate,smooth}
\begin{equation}
    \left\|h\left(\left(|\phi_1|^2*w\right) \phi_2\right)\right\|_{L^2} \lesssim \left\|h^{1/2}\phi_1\right\|_{L^2}^2\left\|h\phi_2\right\|_{L^2}
\end{equation}
\end{proposition}
\begin{proof}
A direct computation shows
\begin{align*}
    h\left(\left(|\phi_1|^2*w\right)\phi_2\right) &=-\Delta\left(|\phi_1|^2*w\right)\phi_2+\left(|\phi_1|^2*w\right)h\phi_2\\
    &\quad +\underbrace{\left(-2\partial_{\bar{z}}\left(|\phi_1|^2*w\right)\right)D^*\phi_2+\left(2\partial_{z}\left(|\phi_1|^2*w\right)\right)D\phi_2-2\partial_{x^3}\left(|\phi_1|^2*w\right)\partial_{x^3}\phi_2}_{\text{first-order terms}}.
\end{align*}
By the proof of Proposition \ref{estimate,local,energy,level}, 
\begin{equation*}
    \left\|\left(|\phi_1|^2*w\right)h\phi_2\right\|_{L^2}\le\left\||\phi_1|^2*w\right\|_{L^{\infty}}\left\|h\phi_2\right\|_{L^2} \lesssim\left\|h^{1/2}\phi_1\right\|_{L^2}^2\left\|h\phi_2\right\|_{L^2},
\end{equation*}
and
\begin{align*}
    \left\|\text{first-order terms}\right\|_{L^2}&\lesssim\left\||\nabla|\left(|\phi_1|^2*w\right)\right\|_{L^3}\left(\|D^*\phi_2\|_{L^6}+\|D\phi_2\|_{L^6}+\|\partial_{x^3}\phi_2\|_{L^6}\right)\\
        &\lesssim\left\|h^{1/2}\phi_1\right\|_{L^2}^2\left\|h\phi_2\right\|_{L^2}.
\end{align*}
Analyzing $-\Delta\left(|\phi_1|^2*w\right)\phi_2$, by the Hardy-Littlewood-Sobolev inequality and the Sobolev inequality,
\begin{align*}
    \left\|-\Delta\left(|\phi_1|^2*w\right)\phi_2\right\|_{L^2}&\le \left\|\left(|\nabla||\phi_1|^2\right)*(|\nabla|w)\right\|_{L^3}\left\|\phi_2\right\|_{L^6}\\
    &\lesssim \left\|\nabla|\phi_1|^2\right\|_{L^{3/2}}\left\|\phi_2\right\|_{L^6}\\
    &\lesssim \left\|\nabla\left|\phi_1\right|\right\|_{L^2}\left\|\phi_1\right\|_{L^6}\left\|\phi_2\right\|_{L^6}\\
    &\lesssim \left\|h^{1/2}\phi_1\right\|_{L^2}^2\left\|h^{1/2}\phi_2\right\|_{L^2}.
\end{align*}
\end{proof}

Using the same argument in Proposition \ref{estimate,operator,local,energy,level}, we generalize Proposition \ref{estimate,smooth} to operators.
\begin{proposition}\label{estimate,smooth,operator}
Suppose $\Gamma_1$ and $\Gamma_{2}$ are self-adjoint, 
\begin{equation}
    \left\|\left[\rho_{\Gamma_{1}}*w,\Gamma_{2}\right]\right\|_{\mathcal{L}^{2,1}}\lesssim \left\|\Gamma_{1}\right\|_{\mathcal{L}^{1,1}}\left\|\Gamma_{2}\right\|_{\mathcal{L}^{2,1}}.
\end{equation}
\end{proposition}

\begin{theorem}\label{local,smooth}
For any initial data $\left\|\Gamma_{0}\right\|_{\mathcal{L}^{2,1}}<\infty$ and $\Gamma_{0}^*=\Gamma_{0}$, Equation (\ref{hatree,density,form}) has a mild solution in the Banach space $\mathbf{N}_{2T}$, where the norm $\mathbf{N}_{2T}$ is defined as 
\begin{equation}
    \left\| \Gamma(t)\right\|_{\mathbf{N}_{2T}} :=\left\| \Gamma(t)\right\|_{L^{\infty}\left([0,T];\mathcal{L}^{2,1}\right)},
\end{equation}
while $I_{T}=[0,T]$ and the existence time $T$ depends on $\left\|\Gamma_{0}\right\|_{\mathcal{L}^{1,1}}$. More precisely, the solution $\Gamma(t)\in C^0\left([0,T],\mathcal{L}^{2,1}\right)\cap C^1\left([0,T],\mathcal{L}^{0,1}\right)$
\end{theorem}
\begin{proof}
Based on Proposition \ref{estimate,smooth,operator}, we use the contraction mapping principle to obtain the local well-posedness result. 

To show the existence time $T$ depends on $\left\|\Gamma_{0}\right\|_{\mathcal{L}^{1,1}}$, consider the integral form of the solution $\Gamma(t)$
\begin{equation*}
    \Gamma(t)=e^{-i\,ht}\Gamma_{0}e^{i\,ht}-i\int_{0}^{t}e^{-i\,h(t-\tau)} \left[\rho_{\Gamma(\tau)}*w,\Gamma(\tau)\right]e^{i\,h(t-\tau)} \,d\tau,
\end{equation*}
then by the Minkowski's inequality, 
\begin{align*}
    \left\|\Gamma(t)\right\|_{\mathcal{L}^{2,1}}&\le \left\|e^{-i\,ht}\Gamma_{0}e^{i\,ht}\right\|_{\mathcal{L}^{2,1}}+\int_{0}^{t}\left\|\left[\rho_{\Gamma(\tau)}*w,\Gamma(\tau)\right]\right\|_{\mathcal{L}^{2,1}}\,d\tau\\
    &\le \left\|\Gamma_{0}\right\|_{\mathcal{L}^{2,1}}+C\left(\sup_{\tau\in I_{T}}\left\|\Gamma(\tau)\right\|_{\mathcal{L}^{1,1}}\right)\int_{0}^{t}\left\|\Gamma(\tau)\right\|_{\mathcal{L}^{2,1}}\,d\tau \qquad (\text{Proposition \ref{estimate,smooth,operator}}),
\end{align*}
where $C$ is a constant. Using the Gr\"onwall's inequality, for $0\le t\le T$,
\begin{equation*}
    \left\|\Gamma(t)\right\|_{\mathcal{L}^{2,1}}\le \left\|\Gamma_{0}\right\|_{\mathcal{L}^{2,1}}\exp\left(Ct\sup_{\tau\in I_{T}}\left\|\Gamma(\tau)\right\|_{\mathcal{L}^{1,1}}\right).
\end{equation*}
Since Theorem \ref{local,energy,level} says that the existence $T$ depends on $\left\|\Gamma_{0}\right\|_{\mathcal{L}^{1,1}}$, with the above estimate, so is the case for Theorem \ref{local,smooth}. By the semi-group theory, the solution $\Gamma(t)\in C^0\left([0,T],\mathcal{L}^{2,1}\right)\cap C^1\left([0,T],\mathcal{L}^{0,1}\right)$.
\end{proof}

3. The conservation law.

We first verify the conservation law of energy for smooth data, then pass it to the energy level by the limiting argument.
\begin{proposition}\label{energy,conservation,smooth}
Suppose that $\Gamma(t)\in C^0\left([0,T],\mathcal{L}^{2,1}\right)\cap C^1\left([0,T],\mathcal{L}^{0,1}\right)$ is a solution to Equation (\ref{hatree,density,form}), then the total energy (\ref{total,energy}) $\mathcal{E}_{HF}\left(\Gamma(t)\right)$ is conserved for $t\in [0,T]$.
\end{proposition}
\begin{proof}
The trick is to express (\ref{total,energy}) in the following way
\begin{align*}
    \mathcal{E}_{HF}\left(\Gamma\right)&=Tr\left(h\Gamma\right)+\frac{1}{2}Tr\left(\left(\rho_{\Gamma}*w\right)\Gamma\right)=Tr\left(\Gamma h\right)+\frac{1}{2}Tr\left(\Gamma\left(\rho_{\Gamma}*w\right)\right),
\end{align*}
and use the mild formulation
\begin{equation*}
    \Gamma(t)=e^{-ih t}\Gamma_{0}e^{iht}-i\int_{0}^{t}e^{-ih(t-\tau)}\left[\rho_{\Gamma(\tau)}*w,\Gamma(\tau)\right]e^{ih(t-\tau)}\,d\tau.
\end{equation*}
Taking the time derivative
\begin{align*}
    \frac{d\,\mathcal{E}_{HF}(\Gamma(t))}{dt} &=-i\,Tr\left(h e^{-ih t}\Gamma_{0}e^{iht}h \right)-\int_{0}^{t}d\tau\, Tr\left(h e^{-ih(t-\tau)}\left[\rho_{\Gamma(\tau)}*w,\Gamma(\tau)\right]e^{ih(t-\tau)}h\right)\\
    &\quad +i\,Tr\left(h e^{-ih t}\Gamma_{0}e^{iht}h \right)+\int_{0}^{t}d\tau\, Tr\left(h e^{-ih(t-\tau)}\left[\rho_{\Gamma(\tau)}*w,\Gamma(\tau)\right]e^{ih(t-\tau)}h\right)\\
    &\quad -i\,Tr\left(\left[\rho_{\Gamma(t)}*w,\Gamma(t)\right]h\right)+Tr\left(\dot{\Gamma}(t)\left(\rho_{\Gamma(t)}*w\right)\right)\\
    &=-i\,Tr\left(\left[\rho_{\Gamma(t)}*w,\Gamma(t)\right]h\right)-i\,Tr\left(\left[h+\rho_{\Gamma(t)}*w,\Gamma(t)\right]\left(\rho_{\Gamma(t)}*w\right)\right)\\
    &=0 \qquad (\text{cyclicity of Tr}).
\end{align*}
By the fundamental theorem of calculus, $\mathcal{E}_{HF}\left(\Gamma(t)\right)=\mathcal{E}_{HF}\left(\Gamma_{0}\right)$ for $0\le t\le T$.
\end{proof}

For any initial data $\Gamma_{0}$ at the energy level, i.e.
\[
    \left\|\Gamma_{0}\right\|_{\mathcal{L}^{1,1}}<\infty,\quad \Gamma_{0}^*=\Gamma_{0},
\]
there exists a sequence $\{\Gamma_{0,k}\}_{k=1}^{\infty}\subset \mathcal{L}^{2,1}$ such that 
\[
    \lim_{k\rightarrow\infty}\left\|\Gamma_{0,k}-\Gamma_{0}\right\|_{\mathcal{L}^{1,1}}=0.
\]
Denote the solution of Equation (\ref{hatree,density,form}) associated to the initial data $\Gamma_{0,k}$ by $\Gamma_{k}(t)$. Since the existence time of $\Gamma_{k}(t)$ depends on $\left\|\Gamma_{0,k}\right\|_{\mathcal{L}^{1,1}}$ (Theorem \ref{local,smooth}), there is a uniform time $T$ such that all solutions $\Gamma_{k}(t)$ exist in the sense of Theorem \ref{local,smooth}.  By the continuous dependence on initial data (from Theorem \ref{local,energy,level}), for any $0\le t \le T$,
\[
    \lim_{k\rightarrow\infty}\left\|\Gamma_{k}(t)-\Gamma(t)\right\|_{\mathcal{L}^{1,1}}=0.
\]
While the total energy $\mathcal{E}_{HF}$ is continuous with respect to the norm $\mathcal{L}^{1,1}$, by Proposition \ref{energy,conservation,smooth},
\begin{equation}\label{energy,conservation}
    \mathcal{E}_{HF}\left(\Gamma(t)\right)=\lim_{k\rightarrow\infty}\mathcal{E}_{HF}\left(\Gamma_{k}(t)\right)=\lim_{k\rightarrow\infty}\mathcal{E}_{HF}\left(\Gamma_{0,k}\right)=\mathcal{E}_{HF}\left(\Gamma_0\right).
\end{equation}

4. The global well-posedness at the energy level.

Note that when the initial data $\Gamma_{0}$ is non-negative, i.e. it satisfies the operator inequality $\Gamma_{0}\ge 0$, the condition of being non-negative is preserved under Equation (\ref{hatree,density,form}). Thus $Tr\left(h^{1/2}\Gamma(t)h^{1/2}\right)=\left\|\Gamma(t)\right\|_{\mathcal{L}^{1,1}}$ and the energy $\mathcal{E}_{HF}(\Gamma(t))\sim \left\|\Gamma(t)\right\|_{\mathcal{L}^{1,1}}$. Using the conservation law (\ref{energy,conservation}), we improve the local well-posedness result Theorem \ref{local,energy,level} to the following global statement.
\begin{theorem}
Suppose that the initial data $\Gamma_{0}$ satisfies
\[
    \left\|\Gamma_{0}\right\|_{\mathcal{L}^{1,1}}<\infty,\quad \Gamma_{0}^*=\Gamma_{0},\quad \Gamma_{0}\ge 0,
\]
then Equation (\ref{hatree,density,form}) has a global mild solution $\Gamma(t)\in C^{0}\left([0,\infty),\mathcal{L}^{1,1}\right)$.
\end{theorem}

\bibliographystyle{amsalpha}
\bibliography{reference}

\providecommand{\bysame}{\leavevmode\hbox to3em{\hrulefill}\thinspace}
\providecommand{\MR}{\relax\ifhmode\unskip\space\fi MR }
\providecommand{\MRhref}[2]{%
  \href{http://www.ams.org/mathscinet-getitem?mr=#1}{#2}
}
\providecommand{\href}[2]{#2}
\begin{thebibliography}{BGGM03}

\bibitem[BA10]{Ali2010}
Besma Ben~Ali, \emph{Maximal inequalities and {R}iesz transform estimates on
  {$L^p$} spaces for magnetic {S}chr\"{o}dinger operators {I}}, J. Funct. Anal.
  \textbf{259} (2010), no.~7, 1631--1672. \MR{2665406}

\bibitem[BDPF74]{BPF1974}
A.~Bove, G.~Da~Prato, and G.~Fano, \emph{An existence proof for the
  {H}artree-{F}ock time-dependent problem with bounded two-body interaction},
  Comm. Math. Phys. \textbf{37} (1974), 183--191. \MR{424069}

\bibitem[BDPF76]{BPF1976}
\bysame, \emph{On the {H}artree-{F}ock time-dependent problem}, Comm. Math.
  Phys. \textbf{49} (1976), no.~1, 25--33. \MR{456066}

\bibitem[BGGM03]{BGGAM2003}
Claude Bardos, Fran\c{c}ois Golse, Alex~D. Gottlieb, and Norbert~J. Mauser,
  \emph{Mean field dynamics of fermions and the time-dependent {H}artree-{F}ock
  equation}, J. Math. Pures Appl. (9) \textbf{82} (2003), no.~6, 665--683.
  \MR{1996777}

\bibitem[BPS14]{BPS2014}
Niels Benedikter, Marcello Porta, and Benjamin Schlein, \emph{Mean-field
  evolution of fermionic systems}, Comm. Math. Phys. \textbf{331} (2014),
  no.~3, 1087--1131. \MR{3248060}

\bibitem[CH16]{CH2016}
Xuwen Chen and Justin Holmer, \emph{Correlation structures, many-body
  scattering processes, and the derivation of the {G}ross-{P}itaevskii
  hierarchy}, Int. Math. Res. Not. IMRN (2016), no.~10, 3051--3110.
  \MR{3551830}

\bibitem[Cha76]{Cha1976}
J.~M. Chadam, \emph{The time-dependent {H}artree-{F}ock equations with
  {C}oulomb two-body interaction}, Comm. Math. Phys. \textbf{46} (1976), no.~2,
  99--104. \MR{411439}

\bibitem[CHP17]{CHP2017}
Thomas Chen, Younghun Hong, and Nata\v{s}a Pavlovi\'{c}, \emph{Global
  well-posedness of the {NLS} system for infinitely many fermions}, Arch.
  Ration. Mech. Anal. \textbf{224} (2017), no.~1, 91--123. \MR{3609246}

\bibitem[CHP18]{CHP2018}
\bysame, \emph{On the scattering problem for infinitely many fermions in
  dimensions {$d\ge3$} at positive temperature}, Ann. Inst. H. Poincar\'{e}
  Anal. Non Lin\'{e}aire \textbf{35} (2018), no.~2, 393--416. \MR{3765547}

\bibitem[EESY04]{EESY04}
Alexander Elgart, L\'{a}szl\'{o} Erd\H{o}s, Benjamin Schlein, and Horng-Tzer
  Yau, \emph{Nonlinear {H}artree equation as the mean field limit of weakly
  coupled fermions}, J. Math. Pures Appl. (9) \textbf{83} (2004), no.~10,
  1241--1273. \MR{2092307}

\bibitem[FK11]{FW11}
J\"{u}rg Fr\"{o}hlich and Antti Knowles, \emph{A microscopic derivation of the
  time-dependent {H}artree-{F}ock equation with {C}oulomb two-body
  interaction}, J. Stat. Phys. \textbf{145} (2011), no.~1, 23--50. \MR{2841931}

\bibitem[Fol89]{Folland89}
Gerald~B. Folland, \emph{Harmonic analysis in phase space}, Annals of
  Mathematics Studies, vol. 122, Princeton University Press, Princeton, NJ,
  1989. \MR{983366}

\bibitem[GM17]{GM17}
M.~Grillakis and M.~Machedon, \emph{Pair excitations and the mean field
  approximation of interacting bosons, {II}}, Comm. Partial Differential
  Equations \textbf{42} (2017), no.~1, 24--67. \MR{3605290}

\bibitem[GV92]{GinibreVelo1992}
J.~Ginibre and G.~Velo, \emph{Smoothing properties and retarded estimates for
  some dispersive evolution equations}, Comm. Math. Phys. \textbf{144} (1992),
  no.~1, 163--188. \MR{1151250}

\bibitem[KM08]{Klainerman2008}
Sergiu Klainerman and Matei Machedon, \emph{On the uniqueness of solutions to
  the {G}ross-{P}itaevskii hierarchy}, Comm. Math. Phys. \textbf{279} (2008),
  no.~1, 169--185. \MR{2377632}

\bibitem[Kra05]{Krasikov05}
Ilia Krasikov, \emph{Inequalities for laguerre polynomials}, East Journal on
  Approximations \textbf{11} (2005), no.~3, 257--268.

\bibitem[Kra07]{Krasikov07}
\bysame, \emph{Inequalities for orthonormal {L}aguerre polynomials}, J. Approx.
  Theory \textbf{144} (2007), no.~1, 1--26. \MR{2287374}

\bibitem[KT98]{KeelTaoEn}
Markus Keel and Terence Tao, \emph{Endpoint {S}trichartz estimates}, Amer. J.
  Math. \textbf{120} (1998), no.~5, 955--980. \MR{1646048}

\bibitem[KZ10]{KraZar2010}
Ilia Krasikov and Alexander Zarkh, \emph{Equioscillatory property of the
  {L}aguerre polynomials}, J. Approx. Theory \textbf{162} (2010), no.~11,
  2021--2047. \MR{2732920}

\bibitem[LL77]{LandauLifVol3}
L.~D. Landau and E.~M. Lifschitz, \emph{Quantum mechanics non-relativistic
  theory : Volume 3 of course of theoretical physics}, 3 ed., Pergamon Press,
  1977.

\bibitem[LS14]{LS2014}
Mathieu Lewin and Julien Sabin, \emph{The {H}artree equation for infinitely
  many particles, {II}: {D}ispersion and scattering in 2{D}}, Anal. PDE
  \textbf{7} (2014), no.~6, 1339--1363. \MR{3270166}

\bibitem[LS15]{LS2015}
\bysame, \emph{The {H}artree equation for infinitely many particles {I}.
  {W}ell-posedness theory}, Comm. Math. Phys. \textbf{334} (2015), no.~1,
  117--170. \MR{3304272}

\bibitem[NS81]{NS1981}
Heide Narnhofer and Geoffrey~L. Sewell, \emph{Vlasov hydrodynamics of a quantum
  mechanical model}, Comm. Math. Phys. \textbf{79} (1981), no.~1, 9--24.
  \MR{609224}

\bibitem[SO96]{MQC96}
Attila Szabo and Neil~S. Oslund, \emph{Modern quantum chemistry: Introduction
  to advanced electronic structure theory}, Dover Publications, Inc. Mineola,
  New York, 1996.

\bibitem[Sol91]{Solovej1991}
Jan~Philip Solovej, \emph{Proof of the ionization conjecture in a reduced
  {H}artree-{F}ock model}, Invent. Math. \textbf{104} (1991), no.~2, 291--311.
  \MR{1098611}

\bibitem[Spo81]{Spohn1981}
H.~Spohn, \emph{On the {V}lasov hierarchy}, Math. Methods Appl. Sci. \textbf{3}
  (1981), no.~4, 445--455. \MR{657065}

\bibitem[Sze75]{Gabor75}
G{\'a}bor Szeg{\H o}, \emph{Orthogonal polynomials}, fourth ed., American
  Mathematical Society, 1975.

\bibitem[Tao09]{Tao09}
Terence Tao, \emph{A pseudoconformal compactification of the nonlinear
  {S}chr\"{o}dinger equation and applications}, New York J. Math. \textbf{15}
  (2009), 265--282. \MR{2530148}

\bibitem[Zag92]{Zag1992}
Sandro Zagatti, \emph{The {C}auchy problem for {H}artree-{F}ock time-dependent
  equations}, Ann. Inst. H. Poincar\'{e} Phys. Th\'{e}or. \textbf{56} (1992),
  no.~4, 357--374. \MR{1175475}

\end{thebibliography}

\end{document}